\documentclass[nofootinbib, prl, reprint, preprintnumbers, superscriptaddress]{revtex4-2}

\usepackage{relsize}

\usepackage[acronyms, nohypertypes={acronym}, nopostdot, style=super, nonumberlist, toc]{glossaries}
\glsenableentrycount
\setacronymstyle{long-sc-short}

\newcommand\ac[1]{\gls{#1}}

\def\acresetall{\glsresetall[acronym]}

\newacronym{WF}{wf}{Wilson-Fisher}
\newacronym{AF}{af}{asymptotically free}
\newacronym{RG}{rg}{renormalization group}
\newacronym{WZW}{wzw}{Wess-Zumino-Witten}
\newacronym[longplural={conformal field theories}]{CFT}{cft}{conformal field theory}
\newacronym[longplural={lattice field theories}]{LFT}{lft}{lattice field theory}
\newacronym[longplural={effective field theories}]{EFT}{eft}{effective field theory}
\newacronym[longplural={quantum field theories}]{QFT}{qft}{quantum field theory}
\newacronym{LEC}{lec}{low-energy constant}
\newacronym{QCD}{qcd}{quantum chromodynamics}
\newacronym{MC}{mc}{Monte Carlo}
\newacronym{IR}{ir}{infrared}
\newacronym{UV}{uv}{ultraviolet}
\newacronym{SNR}{snr}{signal-to-noise ratio}
\newacronym{NLSM}{nl$\sigma$m}{nonlinear sigma model}
\newacronym{PCM}{pcm}{principal chiral model}
\newacronym{CSA}{csa}{Cartan subalgebra}
\newacronym{SSB}{ssb}{spontaneous symmetry breaking}
\newacronym{DOF}{dof}{degrees of freedom}
\newacronym{DMRG}{dmrg}{densiy matrix renormalization group}
\newacronym{YM}{ym}{Yang-Mills}
\newacronym{QLM}{qlm}{quantum link model}
\newacronym{KG}{kg}{Kogut-Susskind}
\newacronym{KG-QLM}{kg-qlm}{Kogut-Susskind quantum link model}
\newacronym{D-QLM}{d-qlm}{D-theory quantum link model}
\newacronym{SPT}{spt}{symmetry protected topological} 
\newacronym{GW}{gw}{Ginsparg-Wilson}
\newacronym{LGW}{lgw}{L\"uscher-Ginsparg-Wilson}
\newacronym{FK}{fk}{Fidkowski-Kitaev}
\newacronym{CS}{cs}{Chern-Simons}
\newacronym{APS}{aps}{Atiyah-Patodi-Singer}
\newacronym{PV}{pv}{Pauli-Villars}
\newacronym{PBC}{pbc}{periodic boundary conditions}
\newacronym{OBC}{obc}{open boundary conditions}
\newacronym{ABC}{abc}{antiperiodic boundary conditions}
\newacronym{KD}{kd}{Kahler-Dirac}
\newacronym{RKD}{rkd}{reduced Kahler-Dirac}
\newacronym{SMG}{smg}{symmetric mass generation}
\newacronym{BZ}{bz}{Brilloiun zone}
\newacronym{QSP}{qsp}{quantum signal processing}

\newacronym{CHN}{chn}{Creutz-Horvath-Neuberger}
\newacronym{GEP}{gep}{generalized eigenvalue problem}

\usepackage{latexsym}
\usepackage{mathtools}
\usepackage{amsmath,amssymb,amsbsy,amsfonts}
\usepackage{microtype}

\usepackage{graphics}
\graphicspath{{./fig/}}

\usepackage[normalem]{ulem}

\usepackage{bbold}

\usepackage{physics}
\usepackage{physics2}
\usephysicsmodule{ab,ab.braket}
\usepackage{tikz}
\usetikzlibrary{3d}

\newcommand\Order{O}

\newcommand\del\partial

\newcommand\tsum{\textstyle\sum}

\newcommand\Id{\mathbb{I}}

\usepackage{xcolor}
\usepackage{colortbl}
\usepackage{tabularx}
\usepackage{booktabs}

\usepackage{makecell}

\newcommand\minisec[1]{\paragraph{{#1}---}}

\usepackage{enumitem}

\usepackage{graphicx}
\usepackage{amsfonts}
\usepackage{tocvsec2}
\usepackage{slashed}
\usepackage{cancel}
\usepackage{comment}

\usepackage{minibox}
\usepackage{orcidlink}

\newcommand\beq{\begin{eqnarray}}
\newcommand\eeq{\end{eqnarray}}

\newcommand\Qa{{Q}_5}
\newcommand\gmod{\hat{\gamma}_5}
\newcommand\gfive{\gamma_5}

\newcommand\SU{\mathrm{SU}}

\usepackage{xparse}

\newcommand\hw{h^{\text{w}}}
\newcommand\hov{h^{\text{ov}}}

\newcommand\Ls{L_5}

\newcommand\hdw{h^\text{dw}}

\usepackage{ninecolors}
\usepackage{hyperref}
\hypersetup{
  unicode,
  pdfprintscaling=None,
  breaklinks,
  colorlinks=true,   %
  linkcolor=azure5,    %
  citecolor=azure5,    %
  filecolor=magenta, %
  urlcolor=blue5,      %
}
\usepackage[capitalise]{cleveref}
\usepackage{booktabs}
\usepackage{multirow}
\crefname{section}{Sec.}{Secs.}

\newcommand\psibf{\pmb{\psi}}

\usetikzlibrary{quantikz2}

\usepackage{amsthm}
\newtheorem{theorem}{Theorem}

\newcommand\PREP{\textsc{prepare}}
\newcommand\SEL{\textsc{select}}

\newcommand\calQ{\mathcal{Q}}
\newcommand\HG{\mathcal{H}_G}
\newcommand\HF{\mathcal{H}_F}

\newcommand\hwhat{\hat{\mathbf{h}}^{\text{w}}}
\newcommand\hwop{\hat{h}^{\text{w}}}
\newcommand\hop{\hat{h}}
\newcommand\hopG{\hat{h}_G}
\newcommand\hwopG{\hat{h}_G}

\newcommand\errSign{\epsilon_e}
\newcommand\errTime{\epsilon_t}

\newcommand\Ub[1]{{U_{#1}}}
\newcommand\PMUb[1]{\mathcal{P}_M(\Ub{#1})}
\newcommand\ehat[1]{\ensuremath{\hat{e}_{#1}}}

\begin{document}

\preprint{FERMILAB-PUB-26-0223-T}

\title{Exact chiral symmetry with quantum signal processing }

\author{Henry Lamm\,\orcidlink{0000-0003-3033-0791}}
\email{hlamm@fnal.gov}
\affiliation{Fermi National Accelerator Laboratory, Batavia, Illinois 60510, USA}
\author{Alessandro Roggero}
\email{a.roggero@unitn.it}
\affiliation{Physics Department, University of Trento, Via Sommarive 14, I-38123 Trento, Italy}

\author{Hersh Singh\,\orcidlink{0000-0002-2002-6959}}
\email{hershs@fnal.gov}
\affiliation{Fermi National Accelerator Laboratory, Batavia, Illinois 60510, USA}

\author{Luca Spagnoli}
\email{luca.spagnoli@unitn.it}
\affiliation{Physics Department, University of Trento, Via Sommarive 14, I-38123 Trento, Italy}

\begin{abstract}
  We give a quantum signal processing (\textsc{qsp}) algorithm for the overlap fermion Hamiltonian which preserves the Ginsparg-Wilson relation up to a controllable error $\epsilon_e$.
  Quantum simulations of Dirac fermions with exact chiral symmetry are thus nearly free: applying the overlap Hamiltonian costs only a factor logarithmic in $\epsilon_e$ more than the Wilson-Dirac Hamiltonian. Comparing to domain-wall fermions, a mild overhead is found in circuit complexity while reducing qubit costs. We show how \textsc{qsp} effectively constructs an extra dimension when simulating the overlap operator, illustrating that the scaling of quantum algorithms reflects the deeper physics of overlap fermions arising at the boundary of domain-wall fermions.
\end{abstract}

\maketitle
\acresetall

\section{Introduction}
Quantum simulation is the next frontier for systematically computing nonperturbative, nonequilibrium observables from lattice \ac{QCD}.
This raises the question of how to best deal with fermion-doubling
\cite{nielsen_absence_1981a,nielsen_absence_1981,karsten_lattice_1981}
and the problem of chiral symmetry in the Hamiltonian formulation.
In the conventional Euclidean spacetime approach, this problem was elegantly solved by the \ac{GW} relation \cite{ginsparg_remnant_1982} and the overlap operator \cite{neuberger_vectorlike_1998,neuberger_more_1998}, and the domain-wall operator \cite{kaplan_method_1992,kaplan_chiral_1993}.
While domain-wall fermions carry over simply to a Hamiltonian setting,
there is no canonical Hamiltonian analogue of overlap fermions
\cite{singh_ginspargwilson_2025,gioia_exact_2025,clancy_ginspargwilson_2024,chatterjee_quantized_2024,yamaoka_quantized_2025,haegeman_interacting_2024,
cheluvaraja_axial_2001,creutz_new_2002}. 

Here we take a different perspective and ask --- what are the most efficient quantum algorithms for
the different formulations of lattice fermions with chiral symmetry?
We take as case study a standard domain-wall fermion and the overlap Hamiltonian proposed by Refs.~\cite{cheluvaraja_axial_2001,creutz_new_2002}. In particular, we will compare the memory and gate complexity of applying the time-evolution generated by the domain-wall and overlap Hamiltonians. The complexity of domain-wall fermions is computed using the recent near-optimal scheme for local Hamiltonians from Ref.~\cite{Rhodes:2024zbr}, while we mostly focus on the implementation and optimization of overlap fermions. The most-used algorithms to perform time-evolutions are: Suzuki--Trotter product formula~\cite{Trotter:1959ytf, Suzuki:1976be} and their higher-order generalizations~\cite{Suzuki:1985, PhysRevX.11.011020}, and qubitization algorithms like Quantum Signal Processing \ac{QSP}~\cite{low2017hamiltoniansimulationuniformspectral, Low_2019,Gilyen:2018khw}. Despite \ac{QSP} being asymptotically optimal, Trotterization remains widely used due to its low memory overhead and competitive practical performance, especially when the Hamiltonian is local. In this work, we will focus on the scaling of circuits, and for this reason we will consider a \ac{QSP} implementation of the time-evolution operator~\cite{low2017hamiltoniansimulationuniformspectral, Berry:2014ivo}. This requires in general deeper circuits than trotterization algorithms, ancilla qubits, and long-range qubit connectivity for lattice field theory~\cite{Rajput:2021khs,Kane:2023jdo,Hariprakash:2023tla,Spagnoli:2024mib,Kiss:2024sep,Rhodes:2024zbr,Kane:2024odt,Spagnoli:2025xvk}.  Despite this, for lattice fermions with chiral symmetry, the sign function $\varepsilon(h_W)$ appearing in the overlap operator is highly nonlocal and must itself be approximated~\cite{singh_ginspargwilson_2025}; 
the classical polynomial-based approaches to this~\cite{borici_computational_2004, brower_mobius_2006} have direct quantum analogues through \ac{QSP}~\cite{Gilyen:2018khw, low2017hamiltoniansimulationuniformspectral}.

We first set up our lattice and notation. After describing the Hamiltonian formulation of the domain-wall and overlap operator, we present a \ac{QSP} based approach to simulating domain-wall and overlap fermions.
We find that the chiral symmetry by the \ac{QSP} algorithm is broken softly, which makes overlap fermions comparable to domain-wall fermions.
The scaling of the algorithms provides a quantum-algorithmic realization of the known
correspondence between the overlap operator and the extra dimension of domain-wall
fermions~\cite{narayanan_infinitely_1993,neuberger_vectorlike_1998,neuberger_more_1998,brower_mobius_2006}.
Our implementation of overlap fermions results in a better asymptotic scaling in the memory cost compared to domain-wall fermions, while having an overhead in the gate count. Additional details of the model, such as the structure and encoding of the gauge fields,
determine the constant-factor overheads for any specific target theory but do not alter
these scalings.

\section{Lattice Hamiltonians with Chiral Symmetry}

Let $\psi_x, {\psi}^{\dagger}_x$ be a $2^{\frac{d+1}{2}}$-component Dirac fermion field in $d$ (odd) spatial dimensions at site $x$ of a square lattice of size $N$ in the spatial dimensions.
Additional quantum numbers (like color or flavor) can be introduced via the associated vector space of dimension $\mathcal{N}_I$ and promoting the fermion fields $\psi_x, {\psi}^{\dagger}_x$ to $\mathcal{N}_I2^{\frac{d+1}{2}}$-components.
We denote the link connecting sites $x$ and $x + \hat e_i$ as $(x, i)$, where $\hat e_i$ is a unit vector in the spatial direction $i$.
We use bold $\pmb \psi, \pmb {\psi}^{\dagger}$ to denote the $\mathcal{Q}$-component vector with all indices suppressed where $\mathcal{Q}=N^d \mathcal{N}_I2^{\frac{d+1}{2}}$.
For a gauge group $G$, each link $(x,i)$ carries a Hilbert space $\HG^{(x,i)}$ spanned by the group-element basis $\{\ket|g> : g \in G\}$. On it acts the link operator $U_{(x,i)}$, the parallel transporter from site $x$ to $x+\hat e_i$. In this work, we focus on the fermionic side and thus eave the qubit encoding of link space unspecified, treating \(U_{(x,i)}\) as an abstract operator throughout.
Under a gauge transformation parametrized by the group element $g_{x}$ at site $x$, the fields transform as
\begin{align}
  \psi_{x} &\to g_{x} \psi_x, \quad 
  U_{(x,\hat{i})} \to g_{x+\ehat{i}} U_{( x, i)} {g}^{\dagger}_{x} 
\end{align}
which ensures that ${\psi}^{\dagger}_{x+\ehat{i}}U_{(x,i)} \psi_x$ is invariant.
We denote the Fock space of fermions as $\HF$ with $\dim \HF = 2^\calQ$, and the Hilbert space of all links as $\HG = \bigotimes_{(x,i)} \HG^{(x,i)}$.

Gauge invariance means that the physical Hilbert space is invariant under the above transformations.
Let $D_i, D^*_i$ be the forward and backward covariant derivatives along the direction $i$, which act as
\begin{equation}
\begin{split}
	(D_i\psibf)_{x} &= {U}^{\dagger}_{(x,i)} \psi_{x+\hat{i}} - \psi_{x}\;,\\
    (D^*_i\psibf)_{x} &= \psi_{x} - U_{(x-\hat{i},i)} \psi_{x-\hat{i}}\;.
\end{split}
\end{equation}
Now we define the symmetric-difference operator $\delta_i$ and the lattice Laplacian $\Delta =  \sum_{i=1}^d \Delta_i$ where
\begin{align}
	\delta_i = \frac{1}{2} (D_i + D^*_i), \quad \Delta_i = D_i - D^*_i\;.
\end{align}
Moreover, we write a second-quantized quadratic fermionic Hamiltonian as
\begin{align}
  H = {\psibf}^{\dagger} h \psibf\;,
\end{align}
where $h$ is called the single-particle Hamiltonian.
We use the notation $H^\text{f} = {\psibf}^{\dagger} h^\text{f} \psibf$, where $\text{f}=\text{w}, \text{dw}, \text{ov}$ denotes Wilson-Dirac, Domain-wall or overlap fermions.\footnote{
  We summarize the notation for the fermionic Hamiltonian objects used in this work. 
Lowercase \(h\) denotes single-particle Hamiltonians, while uppercase \(H\) denotes the quadratic second-quantized ones. 
The operator \(\hop = \sum_{ab} h_{ab}\ketbra|a><b|\) represents \(h\) acting on a \(\mathcal{Q}\)-dimensional qudit register without gauge fields, while 
\(\hopG = \sum_{ab} h_{ab}\ketbra|a><b|\otimes U_{ab}\) denotes its gauge-coupled counterpart, with \(U_{ab}\) acting on the link Hilbert space \(\mathcal{H}_G\).
}

The Wilson-Dirac Hamiltonian is given by taking $h = \hw$ with%
\footnote{\emph{Notation:} We use the mostly-minus metric convention, with \(\gamma^0\) Hermitian and the spatial \(\gamma^i\) anti-Hermitian for \(i\neq 0\). 
Defining \(\Gamma^0=\gamma^0\) and \(\Gamma^i=\gamma^0\gamma^i\), all \(\Gamma^\mu\) are Hermitian, mutually anticommuting, and satisfy \((\Gamma^\mu)^2=1\). 
The chirality operator is defined as $\gamma_5=i^{d(d+1)/2}\Gamma^0\Gamma^1\cdots\Gamma^d$,
with \(\gamma_5^\dagger=\gamma_5\) and \(\gamma_5^2=1\).}
\begin{align}
\hw = 
    \tsum_{k=1}^{d} i \Gamma^k \delta_k + \Gamma^0 (m - r \Delta / 2).
    \label{eq:hw_definition}
\end{align}
where $m$ is the bare mass and $r$ is the Wilson parameter.
Even for massless Dirac fermions, Wilson fermions break chiral symmetry at the lattice level, restored only in the continuum limit.
Domain-wall fermions and overlap fermions both remedy this problem.
Domain-wall fermions, $h=\hdw$ are obtained by adding an extra-dimension with \ac{PBC} along the $i=1, \dots, d$ dimension and \ac{OBC} along the $(d+1)th$ bulk dimension.%
\footnote{We use \ac{PBC} throughout and comment on others at the end.}
The domain-wall Hamiltonian is the Wilson-Dirac single-particle Hamiltonian in $d+1$ dimensions, where the extra $(d{+}1)$th direction carries the Dirac structure $\Gamma^{d+1} = i\,\Gamma^0 \gamma_5$, i.e.\ the lower-dimensional chirality $\gamma_5$ acts as the gamma matrix along the extra dimension. The gauge links along the extra dimension are taken to be $U_{(x,d+1)} = 1$ and the links parallel to the boundary are merely copies of boundary links $U_{x, i} = U_{x + \ehat{d+1}, i}$.
For $0 < m < 2r$, the free theory hosts localized Weyl fermions at both the boundaries, which combine to yield a single light Dirac fermion.

On the other hand, overlap fermions are formulated in $d$ spatial dimensions, trading the extra dimension for all-to-all interactions.
A Hamiltonian formulation of overlap fermions was suggested by \cite{creutz_new_2002}, with the single-particle Hamiltonian
\begin{align}
\hov &= \Gamma^0 + \varepsilon(\hw) \label{eq:hov}
\end{align}
with $\varepsilon(X) = X / \sqrt{{X}^{\dagger} X}$ the sign function.
Choosing $0 < m < 2r $ results in a single massless Dirac fermion, with an exact chiral symmetry on the lattice given by $\Qa = \pmb {\psi}^{\dagger} \gmod \pmb \psi $, where $\gmod$ is a ``modified'' $\gamma_5$ operator
\begin{align}
	\gmod = \frac{1}{2} \gamma_5(1 - \Gamma^0 \varepsilon(\hw)).
\end{align}

\section{Algorithms for quantum simulation}

We represent the fermionic degrees of freedom by encoding the $\calQ$ spinless fermions with a register of qubits of size $\calQ$ using the Jordan-Wigner mapping.
We denote as $\Psi_a$ and $\Psi^\dagger_a$, for $a=0,\dotsc,\calQ-1$, the individual spinless fermionic operators forming $\pmb \psi$ and $\pmb {\psi}^{\dagger}$.
With this notation, a quadratic fermionic Hamiltonian can be written as
\begin{equation}
\label{eq:snd_comp}
H = \sum_{a,b=0}^{\calQ-1} h_{ab} U_{ab} \Psi^\dagger_a \Psi_b \;,
\end{equation}
where $h_{ab}$ is a $\calQ\times\calQ$ matrix of numbers while $U_{ab}$ are operators acting on the Hilbert space $\mathcal{H}_G$ of the link variables.%
\footnote{For nearest-neighbor couplings, $U_{ab}$ is simply the link field operator connecting the sites. For couplings beyond nearest-neighbor, $U_{ab}$ involves sums of products of link operators on a path connecting the two sites (Wilson lines).}
The combined indices $a,b$ unwrap into the color $n$, spin $\alpha$, spatial indices $\vec x$ as $a \equiv (n, \alpha, \vec x)$, where $\vec x = (x_1, \dotsc, x_d)$ are the coordinates of a lattice site.

\minisec{Block encoding the second-quantized fermionic Hamiltonian}
$H$ in \cref{eq:snd_comp} acts on the full vector space $\HF \otimes \HG$.
We desire an efficient algorithm to block encode $H^{\text{ov}}$, and use it to time evolve via \ac{QSP}.

A block encoding implements a nonunitary operator $H$ by embedding it as the top-left block of a larger unitary operator $W$ acting on an extended Hilbert space with $b$ ancilla qubits, $\mathbb{C}^{2^b}\otimes\mathcal{H}_F\otimes\mathcal{H}_G$. Concretely, if we set the ancilla to $\ket| 0^{\otimes b} > $, apply $W$, and project back onto $\ket| 0^{\otimes b} >$, then the effective operation on the original system is
\begin{equation}
\label{eq:benc_fham}
\langle 0^{\otimes b}\lvert W\rvert 0^{\otimes b}\rangle \approx \frac{H}{\Lambda}\;,
\end{equation}
with a normalization $\Lambda$, and the approximation holds up to a fixed maximum error in operator norm.

In order to efficently block encode a generic quadratic fermionic Hamiltonian, we start from a single-particle Hamiltonian $h$.
Key to this is the \textit{promotion} of the single-particle Hamiltonian $h$: rather than treating it as a $\calQ\times \calQ$ matrix, we promote it to an operator acting on $q \geq \log_2(\calQ) $ qubits.  This exponentially compresses the Hilbert space with respect to $h$.
The fermionic Fock register itself still requires $\calQ$ qubits, and the controlled
fermionic operations still cost $O(\calQ)$ gates; the compression applies only
to the single-particle structure. We define
\begin{align}
\label{eq:hbf_def}
\hwopG &= 
         \sum_{a,b=0}^{\calQ-1} h_{ab}\rvert a\rangle\langle b\lvert\otimes U_{ab}\;,
\end{align}
where $\ket|a>$ of the auxiliary register tracks the combined indices $a \equiv (n, \alpha, \vec x)$ 
\begin{align}
\ket|a> = \ket|n> \ket|\alpha> \ket| x_1 > \cdots \ket| x_d >
    \label{eq:aux_reg_convention}
\end{align}

Then, we can block encode the rescaled $H/2^q$ with $q\geq\log_2(\calQ)$, using a block-encoding of $h$ together with additional $O(\calQ)$ gates and $O(\log(\calQ))$ auxiliary qubits (the explicit expression is found in the supplementary material). Next, we discuss how to block encode $h$.

Let $A$ be the unitary incrementer on an $N$-dimensional register 
$A = \sum_{j=0}^{N-1}\ketbra| j \oplus 1 > < j |$
where $\oplus$ denotes addition modulo $N$.
We will call $A_i$ the incrementer applied on the spatial component of $\ket|a>$ in direction $i=1,\dots,d$
\begin{align}
	A_i \ket|n, \alpha, x_1, \dots, x_d >
  &= \ket|n, \alpha, x_1, \dots, x_i\oplus 1, \dots, x_d > 
\end{align}
With this, the Wilson-Dirac Hamiltonian operator without gauge fields acting on $q \geq \log_2(\calQ) $ qubits is:
\begin{align}
  \hwop &= 
          \frac{i}{2} \sum_{i=1}^{d} \Gamma^i \left( A_i - A_i^\dagger \right)
  \notag\\&\quad
  - \frac{r}{2}\Gamma^0 \sum_{i=1}^{d} \left( A_i + A_i^\dagger \right) + (m+dr) \Gamma^0 \mathbb{1}\;,
\end{align}
where the $\Gamma$ matrices are unitary operators acting on the spin register.
This is now a linear combination of $4d+1$ unitary operators, and can be block encoded in a standard manner~\cite{childs2012hamiltonian}.
To do so, we introduce $l = \lceil \log_{2} (4d+1) \rceil$, and define a \PREP{} operator $P$ which creates a superposition state
\begin{equation}
\begin{split}
\label{eq:prepare_state}
    &P \lvert 0^{\otimes l} \rangle \lvert 0 \rangle_{s} = \\
    &= \frac{1}{\sqrt{\alpha}}\left( \sum_{k=0}^{2d-1} \frac{1}{\sqrt{2}} \lvert k \rangle + \sum_{k=2d}^{4d-1} \sqrt{\frac{r}{2}} \lvert k \rangle + \sqrt{m+dr} \lvert 4d \rangle\right)\lvert1\rangle_{s} \;, \\
\end{split}
\end{equation}
where we introduced an additional qubit $\lvert\cdot\rangle_{s}$, that will be used to put the $i$ and $-1$ coefficients where needed; below we abbreviate $\lvert 0^{\otimes l}\rangle\lvert 0\rangle_s$ as $\lvert 0^{\otimes l}\rangle$. Then, we define the \SEL{} operator $S$ as
\begin{equation}
    S = \gamma^0 \sum_{k=0}^{2^l-1} \ketbra|k><k| \otimes U_k
\end{equation}
with, for $i = 0, \dots, d-1$,%
\begin{equation}
\label{eq:Ufirstrow}
U_{i} = \gamma^{i+1}A_{i+1}e^{-i\frac{\pi}{2} Z_s}, \quad\
U_{d+i} =\gamma^{i+1} {A}^{\dagger}_{i+1}  e^{i\frac{\pi}{2} Z_s} \;,
\end{equation}
\begin{equation}
\label{eq:Usecondrow}
U_{2d+i} = {A}_{i+1} Z_s , \quad
U_{3d+i} = {A}^{\dagger}_{i+1} Z_s, \quad U_{4d} = \Id,
\end{equation}
and $U_k = \Id$ for $k > 4d$, so that $S$ is unitary.
The operator $S$ acts with the projectors $\ketbra|k><k|$ on the $l$-qubit state prepared by $P$, with the $Z$ rotations of $\mp\pi/2$ applied on the additional qubit $s$, and with $A_i$ and $A_i^\dagger$ on the $q$-qubit auxiliary register used to promote the single particle Hamiltonian to a $q$-qubit operator. Since $P$ leaves the qubit $s$ in $\lvert1\rangle_{s}$, the rotations $e^{\mp i\frac{\pi}{2} Z_s}$ contribute $\pm i$ and $Z_s$ contributes $-1$, reproducing the coefficients of $\hwop$.
Then, the block-encoding of the Wilson Hamiltonian is simply $\Ub{\hwop} = P^\dagger S {P}$, meaning 
\begin{align}
  \braket< 0^{\otimes l} | \Ub{\hwop} | 0^{\otimes l}  >  &= 
  \braket< 0^{\otimes l} | P^\dagger S P | 0^{\otimes l}  >  =  \frac{\hwop}{\alpha} 
\end{align}
where $\alpha = d+2dr+m$ is the $1$-norm of the Hamiltonian
$\hwop$.
The operator $S$ is a sum of the $(4d+1)$ unitaries defined in \cref{eq:Ufirstrow} and \cref{eq:Usecondrow}. The dominant cost for them is coming from the incrementers $A_{i}$, each one of which can be implemented using $O(\log N)$ gates and $O(\log N )$ ancillae following the construction in Ref.~\cite{Gidney_2018}. Adding all of the costs together, the prepare unitary $P$ can be implemented using a general state preparation on $l=O(\log d)$ qubits which can be performed in $O(d)$ gates at most (see e.g.~\cite{Plesch_2011}). For $S$ we instead require $O(d\log N)$ gates and $O(d\log N)$ ancilla qubits, in addition to the $l$ qubits needed to define the state in \cref{eq:prepare_state}.

To add gauge fields, we assume a block encoding $G^{(x,i)}_U$ of the link fields $U_{(x,i)}$ such that
\begin{align}
  \braket< 0^{\otimes b_G} | G_{U}^{(x,i)} | 0^{\otimes b_G} > 
  &= \frac{U_{(x,i)}}{\Lambda_G},
    \label{eq:gu_benc}
\end{align}
for some normalization constant $\Lambda_{G}$ and a fixed register of $b_G$ qubits.
Then, we can slightly modify the \SEL{} operator $S$ into $S_G$, by defining
\begin{align}
  \Sigma_i = \sum_{x} \ketbra| x > < x |  G^{(x,i)}_U, 
\end{align}
and substituting
$A_i \to A_i \Sigma_i,\ A_i^\dagger \to {\Sigma}^{\dagger}_i A_i^\dagger$, which gives
\begin{align}
  \braket< 0^{\otimes (l+b_G)} | {P}^{\dagger} S_G P | 0^{\otimes (l+b_G)}> =
  \frac{\hwopG}{\alpha \Lambda_{G}}\;.
\label{eq:block_encoding_hw}
\end{align}
The total cost will be the one of before, plus the cost of implementing $G^{(x,i)}_U$ on every link of the lattice. Assuming the cost of $G^{(x,i)}_U$ to be $O(1)$ at a fixed truncation of the gauge field, the cost will scale with the number of links, and thus with the volume. By using unary iteration~\cite{PhysRevX.8.041015}, the gate count will scale linearly with the number of $G^{(x,i)}_U$, while the number of ancilla qubits needed will scale with its logarithm. Thus, the total cost is then given by $O(\calQ)$ gates and $O(\log\calQ)$ ancilla qubits.

\minisec{Block encoding the overlap Hamiltonian}
The overlap Hamiltonian requires an implementation of the sign function $\varepsilon(\hw)$ of the single-particle Wilson Hamiltonian $\hw$.
Let $p_M(x)$ be an order-$M$ polynomial approximation to the sign function
such that for a given $\errSign$, we have
\begin{align}
	\ab| \varepsilon(x) - p_M(x) |  \leq \errSign, \quad \forall x \in [-1, -\kappa / 2] \cup [\kappa / 2, 1 ].
    \label{eq:poly_approx_to_sign}
\end{align}
Such a polynomial approximation can be found as long as the order $M$ is chosen large enough (Lemma 10 and Corollary 6 of Ref.~\cite{low2017hamiltoniansimulationuniformspectral})
\begin{align}
	M = \Order\ab(\kappa^{-1} \log \errSign^{-1}),
    \label{eq:poly_order_scaling}
\end{align}
where the parameter $\kappa$ must be chosen as a lower bound on the absolute value of eigenvalues of the block-encoded Wilson Hamiltonian (with gauge fields).
If $\lambda_h$ is a lower bound to the absolute value of eigenvalues of $\hwop$,
the absolute-value of eigenvalues of the block-encoded Hamiltonian are lower-bounded by
$\lambda_{min}=\lambda_h / \alpha \Lambda_G$ and we can therefore choose $\kappa=\lambda_{min}$. 
Note that the free Wilson Hamiltonian $\hwop$ (without gauge fields) is gapped with the smallest absolute eigenvalue $m$. %
However, with dynamical gauge fields the Hamiltonian can develop small eigenvalues. In fact, this is a problem in classical lattice \ac{QCD} calculations, where techniques such as deflation are used to deal with near zero eigenvalues~\cite{Luscher:2007se,Luscher:2007es}.
In this work, we fix $\kappa >0$.

Let $\Ub{\hwop} = {P}^{\dagger} S_G P$ be the block-encoding of the Hamiltonian $\hwop$,
and let $\PMUb{\hwop}$ be the block encoding of the polynomial of order $M$ in the Hamiltonian.
\ac{QSP}, in its quantum singular value transformation formulation, allows a degree-$M$ polynomial transformation of a block-encoded operator to be applied using $M$ alternating applications of $U_{\hat h}$, $U_{\hat h}^\dagger$ and single-qubit rotations on one ancilla.
We take $p_M$ to be the odd polynomial approximation in \cref{eq:poly_approx_to_sign} and can then define
\begin{align}
  E_M = \braket< 0^{\otimes(l+b_G)} | \mathcal{P}_M\ab( U_{\hwop} ) | 0^{\otimes (l+b_G)} >,
\end{align}
as the polynomial approximation to $\varepsilon(\hwopG)$ of order $M$.
Therefore, \cref{eq:poly_approx_to_sign} implies
\begin{align}
  \ab\| \varepsilon(\hw) - E_M(\hw) \| \leq \errSign
\label{eq:BE_polynomial_approx_sign}
\end{align}
with $M$ scaling as given in \cref{eq:poly_order_scaling}.
Moreover,
$\PMUb{\hwop}$ can be implemented with $M$ applications of $\Ub{\hwop}$ and its dagger, plus $O(M\log\calQ)$ gates (Theorem 17 of Ref.~\cite{Gilyen:2018khw}). This means that the cost of implementing the approximate reflection $\varepsilon(\hwop)$ to error less than $\errSign$ scales as
\begin{equation}
    O\ab( \frac{\calQ}{\kappa} \log\frac{1}{\errSign} )\;.
    \label{eq:cost_E_M}
\end{equation}
We note that a direct block encoding of the Wilson Hamiltonian as in \cref{eq:block_encoding_hw}, with an explicitly-broken chiral symmetry, would also have required $O\left( \calQ \right)$ gates but without approximation error or dependence on $\kappa$. We also note that, in general, this scaling with $\kappa$ and the error $\errSign$ is likely to be optimal since an approximation to the sign function can be used as a pre-processing step when solving a linear system by amplifying the spectral gap~\cite{doi:10.1137/120871997,low2017hamiltoniansimulationuniformspectral,x6gh-d8gh}. This expectation follows from the fact that $\kappa^{-1}$ is the condition number of the matrix and there is a known lower bound on the oracle calls to the matrix which scales linearly with the condition number~\cite{PRXQuantum.3.040303,10.1088/2058-9565/ae89e0}.

Putting everything together, we have the single-particle Hamiltonian that can be implemented with $O(\frac{\calQ}{\kappa}\log\frac{1}{\errSign})$ gates and $O(\log(\calQ))$ ancillary qubits. The quadratic fermionic Hamiltonian can be block-encoded as in \cref{eq:benc_fham} using $O(\calQ)$ gates and $O(\log(\calQ))$ ancillary qubits plus the cost of the block-encoding of the single-particle Hamiltonian. Thus, the total cost is
\begin{equation}
    C_H = O\left( \frac{\calQ}{\kappa}\log\frac{1}{\errSign} \right)\,.
\end{equation}
This is only a factor $\kappa^{-1}\log(1/\errSign)$ more than the Wilson case, while preserving chiral symmetry up to $O(\errSign)$.
Once the block-encoding of the Hamiltonian $H$ can be implemented with cost $C_H$, we can use the \ac{QSP} algorithm~\cite{Gilyen:2018khw} and implement a polynomial approximation to the exponential $e^{-iHt}$ with
\begin{align}
  O\ab(C_H \ab( \Lambda_H t + \log{\errTime^{-1}} ))
\end{align}
gates, where $\Lambda_H=O(\calQ)$ is an upper-bound to the $1$-norm of $H$, and $\errTime$ is the maximum allowed error. Considering the overlap Hamiltonian, this means that we can implement its corresponding time-propagator with a gate complexity of
\begin{equation}
    O\ab( \calQ \kappa^{-1} \log\ab( \errSign^{-1} ) \ab( \calQ t + \log{\errTime^{-1}} ))\;.
\end{equation}
In the last equation, $\errSign$ is the error coming from the polynomial approximation of the sign function, while $\errTime$ is the error coming from the polynomial approximation to the exponential function for time evolution.

The cost of Wilson-Dirac, domain-wall and overlap fermions in terms of memory, gate complexity to block-encode the Hamiltonian and the gate complexity to perform time evolution is given in \cref{tab:scaling}.

We note that when considering Wilson and domain-wall fermions, the geometric locality of the Hamiltonian implies a finite Lieb-Robinson velocity~\cite{Haah:2018ekc}, enabling simulation with almost-linear cost $O(\calQ t \text{polylog}(\calQ t / \errTime))$~\cite{Rhodes:2024zbr}.
For overlap fermions, the all-to-all structure of $\hov$ prevents exploiting locality.
This distinction further sharpens the tradeoff: domain-wall fermions are more expensive in qubits but cheaper in gate depth at larger system sizes.

\begin{table*}[tbp]
    \centering
    \renewcommand{\arraystretch}{1.3}
\begin{tabular}{@{}ll@{\hspace{2.2em}}|ccc@{}}
    \toprule
    \toprule
    & & \textbf{Overlap} & \textbf{Wilson} & \textbf{Domain Wall} \\
    \midrule

    \multicolumn{2}{@{}l|}{\textit{Qubit cost}}
    & $O(\calQ+\log\calQ)$
    & $O(\calQ+\log\calQ)$
    & $O(\calQ L_5+\log(\calQ L_5))$ \\

    \addlinespace[0.4em]
    \midrule
    \addlinespace[0.2em]

    \multirow{2}{*}{\makecell[l]{\textit{Circuit}\\\textit{complexity}}}
    & Block enc. of $h$
    & $O(\calQ \kappa^{-1} \log(1/\errSign))$
    & $O(\calQ)$
    & $O(\calQ L_5)$ \\

    & Time evolution
    & $O\ab(\calQ \kappa^{-1}\log\ab(1/\errSign)\ab[\calQ t+\log(1/\errTime)])$
    & $O\ab(\calQ t\,\mathrm{polylog}(\calQ t/\errTime))$
    & $O\ab(\calQ L_5 t\,\mathrm{polylog}(\calQ L_5 t/\errTime))$ \\

    \addlinespace[0.4em]
    \midrule
    \addlinespace[0.2em]

    \multicolumn{2}{@{}l|}{\textit{Chiral symmetry violation}}
    & $O(\epsilon_e)$
    & $O(1)$
    & $O(e^{-cL_5})$ \\

    \bottomrule
        \bottomrule
\end{tabular}
    \caption{Resource comparison for \ac{QSP}-based simulation of
      overlap, Wilson, and domain-wall fermions coupled to gauge fields.
      Here $\errSign$ is the sign-function approximation error,
      $\errTime$
      is the total simulation error, $t$ is the evolution
      time, and $L_5$ is the extent of the extra dimension.
      The identification
      $\epsilon_e \sim e^{-c L_5}$ connects the overlap and
      domain-wall formulations.
      Subleading polylogarithmic factors are suppressed. For Wilson and domain wall the strategies are adapted from Ref.~\cite{Rhodes:2024zbr}, where they are derived for staggered fermions.
      \label{tab:scaling}}
\end{table*}

\section{Exact chiral symmetry and the Ginsparg-Wilson relation}
Consider \cref{eq:BE_polynomial_approx_sign}. By truncating the polynomial to order $M$ we also have prevented the block encoded approximation to $\varepsilon(\hwhat)$ to be unitary.
As discussed above, $E_M$ is Hermitian because $\mathcal{P}_M$ is a real polynomial of definite parity applied to the Hermitian $\hwopG$, so the only error comes from approximating the sign function by a degree $M$ polynomial.
By using \cref{eq:BE_polynomial_approx_sign} together with the fact that $\varepsilon(\hwhat)$ squares to the identity, we can say that
\begin{equation}
\norm{ E_M^2 - \mathbb{1} } \leq 2 \errSign \;.
\end{equation}
The violation of the Ginsparg-Wilson relation is then of the order of $\epsilon_e$.
Let $V = \gamma^0 \varepsilon(\hw)$, and let $\hat{V} = \gamma^0 E_M$ be our approximation to it. For the Dirac operator $D = 1 + V$, the \ac{GW} relation $D + {D}^{\dagger} = {D}^{\dagger} D$ is exact if ${V}^{\dagger}V =1 $. However, with the approximation of the sign function, $\hat{V}$ ceases to be exactly unitary, but rather $\norm{{\hat{V}}^{\dagger} \hat{V} - \mathbf{1} } \leq 2 \epsilon_e  $.  This translates to an error in the \ac{GW} relation,
\begin{align}
\norm{\hat{D} + {\hat{D}}^{\dagger} - {\hat{D}}^{\dagger} \hat{D}} \leq 2 \epsilon_e \;.
\end{align}
The \ac{GW} relation guarantees a chiral symmetry:
\begin{align}
  \norm{ [\gmod, \hat{h} ]  } \leq 2\epsilon_e,
  \label{eq:GW_rel_violation}
\end{align}
where the $\gmod$ is the modified $\gfive$ operator, defined by
\begin{align}
  \gmod = \frac{1}{2} \gfive (1 - \gamma^0 E_M).
  \label{eq:modified_gamma_5}
\end{align}

This QSP-implementable operator is the
analogue of the standard overlap chiral operator $\gamma_5(1-\tfrac12 D)$ and commutes with the exact overlap Hamiltonian. \cref{eq:GW_rel_violation} is derived in the Supplementary Material.

Since the \ac{QSP} construction already provides a block encoding of $E_M$, the modified chiral operator can be implemented with a gate complexity given by \cref{eq:cost_E_M}, which is the same asymptotic cost for applying the single-particle overlap Hamiltonian.

\minisec{Overlap operator and the extra dimension}
The scaling in \cref{tab:scaling} suggests that we can identify $\epsilon_e \sim \exp( -c \Ls) $ for some $c$. 
Indeed, the fact that the domain-wall operator with finite extra dimension $\Ls$ is an approximation to the overlap operator has been known since the discovery of the overlap operator \cite{narayanan_infinitely_1993,neuberger_more_1998}. A Euclidean derivation of the overlap operator as the effective boundary operator of a domain-wall system shows that the domain-wall systems approximates the sign function $\varepsilon(\hw) \sim \tanh(\hw \Ls)$, which becomes exact as $\Ls \to \infty$~\cite{narayanan_infinitely_1993}.
We find this same behavior for the \ac{QSP} implementation of the overlap operator. The approximation of the sign function effectively generates an extra dimension of size $\Ls \sim \kappa^{-1}\log(\epsilon_e^{-1})$.
The same connection is reflected in the sparsity of the approximate overlap operator.
Starting from the $s$-sparse $\hw$ (with $s= \mathcal{O}(d)$ due to nearest-neighbor structure), 
the degree-$M$ polynomial approximation $p_M(\hw)$ is at most $sM$-sparse. Since $M = \mathcal{O}(\kappa^{-1}\log(1/\epsilon_e))$, this gives a sparsity of $\mathcal{O}(d\kappa^{-1} \log(1/\epsilon_e))$, interpolating between the local Wilson operator ($M=1$) 
and the fully nonlocal exact Overlap operator ($M \to \infty$). 
This mirrors the domain-wall picture: $L_5 \sim \kappa^{-1}\log(1/\epsilon_e)$ steps along the fifth dimension spread the support of the effective boundary operator by $L_5$ lattice spacings.

The polynomial approximation of the sign function in the \ac{QSP} construction is closely related to classical L\ac{QCD} approaches. In both cases, achieving an approximation error $\epsilon_e$ over the spectral interval $[\kappa,1]$ requires a degree scaling as $M = O(\kappa^{-1}\log(1/\epsilon_e))$. Classical overlap formulations, however, typically employ optimal rational approximations, such as the Zolotarev approximation~\cite{Chiu:2002eh,vandenEshof:2002ms}, or their M\"obius domain-wall realization~\cite{brower_mobius_2006}, which empirically have been shown to greatly improve the dependence on the spectral gap for the Dirac operator. While these rational approximations are more efficient classically, they have no direct analogue within the standard \ac{QSP} framework. Extending quantum signal processing to efficiently realize rational approximations would therefore be an interesting direction for future work.

We presented \ac{QSP}-based quantum simulation algorithms for overlap fermions and compared them to domain-wall implementations; both preserve an exact chiral symmetry.
For overlap fermions, approximating the sign function with a degree-$M$ polynomial yields gate complexity $O(\mathcal{Q}\kappa^{-1} \log(1/\epsilon_e))$ and violates the \ac{GW} relation only at $O(\epsilon_e)$. The central insight is that \ac{QSP} naturally generates structure analogous to the extra dimension in domain-wall fermions.
The polynomial degree $M \sim\kappa^{-1} \log(1/\epsilon_e)$ plays the same role as $L_5$: both control chiral symmetry violation, both introduce identical logarithmic cost, and both interpolate between local and fully nonlocal operators.
This equivalence reflects the deep connection between overlap fermions and the boundary theory of domain-wall fermions~\cite{narayanan_infinitely_1993,neuberger_more_1998}. The two formulations offer a direct tradeoff: domain-wall requires $\mathcal{Q} L_5$ qubits but retains geometric locality, while overlap requires only $\calQ$ qubits but incurs a worst-case gate count scaling quadratically with system size, $O(\calQ^2)$.
Which formulation is preferable will depend on hardware constraints. The block-encoding framework presented here already accommodates non-Abelian gauge fields through the link operator $G_U^{(x,i)}$ in \cref{eq:gu_benc}. For an $\SU(N_c)$ gauge theory with a truncated link Hilbert space of dimension $\Lambda$, the gauge block encoding requires $b_G = O(\log \Lambda)$ ancilla qubits, while $\Lambda_G$ --- which enters the cost only through $\kappa$ --- is independent of $\Lambda$. In the group-element basis of a finite group or digitized subgroup~\cite{Lamm:2019bik}, $U_{(x,i)} = \sum_g \ketbra|g><g| \otimes D(g)$ is exactly unitary, so $\Lambda_G = 1$ and $b_G = 0$; under a representation-basis truncation $U_{(x,i)}$ couples each irrep to at most $N_c$ others with $O(1)$ matrix elements, so a sparse-access block encoding~\cite{Gilyen:2018khw} gives $\Lambda_G = O(N_c)$. The asymptotic scaling is therefore unaffected. Future work includes obtaining explicit resource estimates for realistic $\SU(3)$ simulations, extending to theories with anomalous chiral symmetry relevant for \ac{QCD}, and benchmarking on near-term devices.

Anti-periodic boundary conditions can be accommodated by replacing $A_i \to \tilde{A}_i$ where $\tilde{A}_i$ differs from $A_i$ by a sign on the wrap-around term, requiring only a single additional multi-controlled-$Z$ gate per direction, which will not change the asymptotic scaling. Open boundary conditions can be implemented in a similar way, modifying the $A_i$ operators. However, this would make $A_i$ not unitary anymore, and this will in general translate in a finite success probability for its implementation, possibly modifying the overall scaling.
For the domain-wall construction, open (Dirichlet) boundary conditions along the extra dimension, $s = 0, \dotsc, L_5 - 1$, are desirable so chiral modes are localized on the two boundaries.

We have quoted costs at fixed $\kappa$; whether a deflation step analogous to the classical one can remove the $\kappa^{-1}$ prefactor on a quantum computer is left to future work. Further, our construction assumed Jordan--Wigner encoding, where the $\mathcal{P}_a$ strings have weight $O(\calQ)$ but the unary-iteration accumulator keeps the total cost at $O(\calQ)$~\cite{PhysRevX.8.041015}; lower-weight encodings~\cite{Bravyi:2000vfj,Havlicek:2017vcq} or geometrically local encodings change constant factors and locality properties and these may prove decisive to the comparison between domain-wall and overlap implementations.

\begin{acknowledgments}
  \emph{Acknowledgments-}
  HS would like to thank Yigal Shamir for an illuminating discussion about treatment of near-zero modes in classical computations of overlap opertor.
  We thank ECT$^*$ and the organizers of the Workshop ``Bridging analytical and numerical methods for quantum field theory'' for their support during the workshop (Trento, August 2025), during which this work was initiated.
  We acknowledge the support of the U.S.\ Department of Energy, Office of
  Science, Office of High Energy Physics Quantum Information Science Enabled
  Discovery (QuantISED) program ``Intersections of QIS and Theoretical Particle Physics''.  This work was produced by Fermi
  Forward Discovery Group, LLC under Contract No.\ 89243024CSC000002 with the
  U.S.\ Department of Energy, Office of Science, Office of High Energy
  Physics.
\end{acknowledgments}

\bibliography{refs.bib}

\begin{thebibliography}{54}%
\makeatletter
\providecommand \@ifxundefined [1]{%
 \@ifx{#1\undefined}
}%
\providecommand \@ifnum [1]{%
 \ifnum #1\expandafter \@firstoftwo
 \else \expandafter \@secondoftwo
 \fi
}%
\providecommand \@ifx [1]{%
 \ifx #1\expandafter \@firstoftwo
 \else \expandafter \@secondoftwo
 \fi
}%
\providecommand \natexlab [1]{#1}%
\providecommand \enquote  [1]{``#1''}%
\providecommand \bibnamefont  [1]{#1}%
\providecommand \bibfnamefont [1]{#1}%
\providecommand \citenamefont [1]{#1}%
\providecommand \href@noop [0]{\@secondoftwo}%
\providecommand \href [0]{\begingroup \@sanitize@url \@href}%
\providecommand \@href[1]{\@@startlink{#1}\@@href}%
\providecommand \@@href[1]{\endgroup#1\@@endlink}%
\providecommand \@sanitize@url [0]{\catcode `\\12\catcode `\$12\catcode
  `\&12\catcode `\#12\catcode `\^12\catcode `\_12\catcode `\%12\relax}%
\providecommand \@@startlink[1]{}%
\providecommand \@@endlink[0]{}%
\providecommand \url  [0]{\begingroup\@sanitize@url \@url }%
\providecommand \@url [1]{\endgroup\@href {#1}{\urlprefix }}%
\providecommand \urlprefix  [0]{URL }%
\providecommand \Eprint [0]{\href }%
\providecommand \doibase [0]{https://doi.org/}%
\providecommand \selectlanguage [0]{\@gobble}%
\providecommand \bibinfo  [0]{\@secondoftwo}%
\providecommand \bibfield  [0]{\@secondoftwo}%
\providecommand \translation [1]{[#1]}%
\providecommand \BibitemOpen [0]{}%
\providecommand \bibitemStop [0]{}%
\providecommand \bibitemNoStop [0]{.\EOS\space}%
\providecommand \EOS [0]{\spacefactor3000\relax}%
\providecommand \BibitemShut  [1]{\csname bibitem#1\endcsname}%
\let\auto@bib@innerbib\@empty
\bibitem [{\citenamefont {Nielsen}\ and\ \citenamefont
  {Ninomiya}(1981{\natexlab{a}})}]{nielsen_absence_1981a}%
  \BibitemOpen
  \bibfield  {author} {\bibinfo {author} {\bibfnamefont {H.~B.}\ \bibnamefont
  {Nielsen}}\ and\ \bibinfo {author} {\bibfnamefont {M.}~\bibnamefont
  {Ninomiya}},\ }\bibfield  {title} {\bibinfo {title} {Absence of neutrinos on
  a lattice: ({{I}}). {{Proof}} by homotopy theory},\ }\href
  {https://doi.org/10.1016/0550-3213(81)90361-8} {\bibfield  {journal}
  {\bibinfo  {journal} {Nuclear Physics B}\ }\textbf {\bibinfo {volume}
  {185}},\ \bibinfo {pages} {20} (\bibinfo {year}
  {1981}{\natexlab{a}})}\BibitemShut {NoStop}%
\bibitem [{\citenamefont {Nielsen}\ and\ \citenamefont
  {Ninomiya}(1981{\natexlab{b}})}]{nielsen_absence_1981}%
  \BibitemOpen
  \bibfield  {author} {\bibinfo {author} {\bibfnamefont {H.~B.}\ \bibnamefont
  {Nielsen}}\ and\ \bibinfo {author} {\bibfnamefont {M.}~\bibnamefont
  {Ninomiya}},\ }\bibfield  {title} {\bibinfo {title} {Absence of neutrinos on
  a lattice: ({{II}}). {{Intuitive}} topological proof},\ }\href
  {https://doi.org/10.1016/0550-3213(81)90524-1} {\bibfield  {journal}
  {\bibinfo  {journal} {Nuclear Physics B}\ }\textbf {\bibinfo {volume}
  {193}},\ \bibinfo {pages} {173} (\bibinfo {year}
  {1981}{\natexlab{b}})}\BibitemShut {NoStop}%
\bibitem [{\citenamefont {Karsten}(1981)}]{karsten_lattice_1981}%
  \BibitemOpen
  \bibfield  {author} {\bibinfo {author} {\bibfnamefont {L.~H.}\ \bibnamefont
  {Karsten}},\ }\bibfield  {title} {\bibinfo {title} {Lattice {{Fermions}} in
  {{Euclidean Space-time}}},\ }\href
  {https://doi.org/10.1016/0370-2693(81)90133-7} {\bibfield  {journal}
  {\bibinfo  {journal} {Phys. Lett. B}\ }\textbf {\bibinfo {volume} {104}},\
  \bibinfo {pages} {315} (\bibinfo {year} {1981})}\BibitemShut {NoStop}%
\bibitem [{\citenamefont {Ginsparg}\ and\ \citenamefont
  {Wilson}(1982)}]{ginsparg_remnant_1982}%
  \BibitemOpen
  \bibfield  {author} {\bibinfo {author} {\bibfnamefont {P.~H.}\ \bibnamefont
  {Ginsparg}}\ and\ \bibinfo {author} {\bibfnamefont {K.~G.}\ \bibnamefont
  {Wilson}},\ }\bibfield  {title} {\bibinfo {title} {A {{Remnant}} of {{Chiral
  Symmetry}} on the {{Lattice}}},\ }\href
  {https://doi.org/10.1103/PhysRevD.25.2649} {\bibfield  {journal} {\bibinfo
  {journal} {Phys. Rev. D}\ }\textbf {\bibinfo {volume} {25}},\ \bibinfo
  {pages} {2649} (\bibinfo {year} {1982})}\BibitemShut {NoStop}%
\bibitem [{\citenamefont
  {Neuberger}(1998{\natexlab{a}})}]{neuberger_vectorlike_1998}%
  \BibitemOpen
  \bibfield  {author} {\bibinfo {author} {\bibfnamefont {H.}~\bibnamefont
  {Neuberger}},\ }\bibfield  {title} {\bibinfo {title} {Vectorlike gauge
  theories with almost massless fermions on the lattice},\ }\href
  {https://doi.org/10.1103/PhysRevD.57.5417} {\bibfield  {journal} {\bibinfo
  {journal} {Phys. Rev. D}\ }\textbf {\bibinfo {volume} {57}},\ \bibinfo
  {pages} {5417} (\bibinfo {year} {1998}{\natexlab{a}})}\BibitemShut {NoStop}%
\bibitem [{\citenamefont
  {Neuberger}(1998{\natexlab{b}})}]{neuberger_more_1998}%
  \BibitemOpen
  \bibfield  {author} {\bibinfo {author} {\bibfnamefont {H.}~\bibnamefont
  {Neuberger}},\ }\bibfield  {title} {\bibinfo {title} {More about exactly
  massless quarks on the lattice},\ }\href
  {https://doi.org/10.1016/S0370-2693(98)00355-4} {\bibfield  {journal}
  {\bibinfo  {journal} {Physics Letters B}\ }\textbf {\bibinfo {volume}
  {427}},\ \bibinfo {pages} {353} (\bibinfo {year} {1998}{\natexlab{b}})},\
  \Eprint {https://arxiv.org/abs/hep-lat/9801031} {arXiv:hep-lat/9801031}
  \BibitemShut {NoStop}%
\bibitem [{\citenamefont {Kaplan}(1992)}]{kaplan_method_1992}%
  \BibitemOpen
  \bibfield  {author} {\bibinfo {author} {\bibfnamefont {D.~B.}\ \bibnamefont
  {Kaplan}},\ }\bibfield  {title} {\bibinfo {title} {A method for simulating
  chiral fermions on the lattice},\ }\href
  {https://doi.org/10.1016/0370-2693(92)91112-M} {\bibfield  {journal}
  {\bibinfo  {journal} {Physics Letters B}\ }\textbf {\bibinfo {volume}
  {288}},\ \bibinfo {pages} {342} (\bibinfo {year} {1992})}\BibitemShut
  {NoStop}%
\bibitem [{\citenamefont {Kaplan}(1993)}]{kaplan_chiral_1993}%
  \BibitemOpen
  \bibfield  {author} {\bibinfo {author} {\bibfnamefont {D.~B.}\ \bibnamefont
  {Kaplan}},\ }\bibfield  {title} {\bibinfo {title} {Chiral fermions on the
  lattice},\ }\href {https://doi.org/10.1016/0920-5632(93)90282-B} {\bibfield
  {journal} {\bibinfo  {journal} {Nuclear Physics B - Proceedings Supplements}\
  }\bibinfo {series} {Proceedings of the {{International Symposium}} On},\
  \textbf {\bibinfo {volume} {30}},\ \bibinfo {pages} {597} (\bibinfo {year}
  {1993})}\BibitemShut {NoStop}%
\bibitem [{\citenamefont {Singh}(2025)}]{singh_ginspargwilson_2025}%
  \BibitemOpen
  \bibfield  {author} {\bibinfo {author} {\bibfnamefont {H.}~\bibnamefont
  {Singh}},\ }\href {https://doi.org/10.48550/arXiv.2505.20419} {\bibinfo
  {title} {Ginsparg-{{Wilson Hamiltonians}} with {{Improved Chiral Symmetry}}}}
  (\bibinfo {year} {2025}),\ \Eprint {https://arxiv.org/abs/2505.20419}
  {arXiv:2505.20419 [hep-lat]} \BibitemShut {NoStop}%
\bibitem [{\citenamefont {Gioia}\ and\ \citenamefont
  {Thorngren}(2025)}]{gioia_exact_2025}%
  \BibitemOpen
  \bibfield  {author} {\bibinfo {author} {\bibfnamefont {L.}~\bibnamefont
  {Gioia}}\ and\ \bibinfo {author} {\bibfnamefont {R.}~\bibnamefont
  {Thorngren}},\ }\href {https://doi.org/10.48550/arXiv.2503.07708} {\bibinfo
  {title} {Exact {{Chiral Symmetries}} of 3+{{1D Hamiltonian Lattice
  Fermions}}}} (\bibinfo {year} {2025}),\ \Eprint
  {https://arxiv.org/abs/2503.07708} {arXiv:2503.07708 [cond-mat]} \BibitemShut
  {NoStop}%
\bibitem [{\citenamefont {Clancy}(2024)}]{clancy_ginspargwilson_2024}%
  \BibitemOpen
  \bibfield  {author} {\bibinfo {author} {\bibfnamefont {M.}~\bibnamefont
  {Clancy}},\ }\bibfield  {title} {\bibinfo {title} {Toward a
  {{Ginsparg-Wilson}} lattice {{Hamiltonian}}},\ }\href
  {https://doi.org/10.1103/PhysRevD.110.L011502} {\bibfield  {journal}
  {\bibinfo  {journal} {Phys. Rev. D}\ }\textbf {\bibinfo {volume} {110}},\
  \bibinfo {pages} {L011502} (\bibinfo {year} {2024})}\BibitemShut {NoStop}%
\bibitem [{\citenamefont {Chatterjee}\ \emph {et~al.}(2024)\citenamefont
  {Chatterjee}, \citenamefont {Pace},\ and\ \citenamefont
  {Shao}}]{chatterjee_quantized_2024}%
  \BibitemOpen
  \bibfield  {author} {\bibinfo {author} {\bibfnamefont {A.}~\bibnamefont
  {Chatterjee}}, \bibinfo {author} {\bibfnamefont {S.~D.}\ \bibnamefont
  {Pace}},\ and\ \bibinfo {author} {\bibfnamefont {S.-H.}\ \bibnamefont
  {Shao}},\ }\href {https://doi.org/10.48550/arXiv.2409.12220} {\bibinfo
  {title} {Quantized axial charge of staggered fermions and the chiral
  anomaly}} (\bibinfo {year} {2024}),\ \Eprint
  {https://arxiv.org/abs/2409.12220} {arXiv:2409.12220} \BibitemShut {NoStop}%
\bibitem [{\citenamefont {Yamaoka}(2025)}]{yamaoka_quantized_2025}%
  \BibitemOpen
  \bibfield  {author} {\bibinfo {author} {\bibfnamefont {T.}~\bibnamefont
  {Yamaoka}},\ }\href {https://doi.org/10.48550/arXiv.2504.10263} {\bibinfo
  {title} {Quantized {{Axial Charge}} in the {{Hamiltonian Approach}} to
  {{Wilson Fermions}}}} (\bibinfo {year} {2025}),\ \Eprint
  {https://arxiv.org/abs/2504.10263} {arXiv:2504.10263 [hep-lat]} \BibitemShut
  {NoStop}%
\bibitem [{\citenamefont {Haegeman}\ \emph {et~al.}(2024)\citenamefont
  {Haegeman}, \citenamefont {Lootens}, \citenamefont {Mortier}, \citenamefont
  {Stottmeister}, \citenamefont {Ueda},\ and\ \citenamefont
  {Verstraete}}]{haegeman_interacting_2024}%
  \BibitemOpen
  \bibfield  {author} {\bibinfo {author} {\bibfnamefont {J.}~\bibnamefont
  {Haegeman}}, \bibinfo {author} {\bibfnamefont {L.}~\bibnamefont {Lootens}},
  \bibinfo {author} {\bibfnamefont {Q.}~\bibnamefont {Mortier}}, \bibinfo
  {author} {\bibfnamefont {A.}~\bibnamefont {Stottmeister}}, \bibinfo {author}
  {\bibfnamefont {A.}~\bibnamefont {Ueda}},\ and\ \bibinfo {author}
  {\bibfnamefont {F.}~\bibnamefont {Verstraete}},\ }\href
  {http://arxiv.org/abs/2405.10285} {\bibinfo {title} {Interacting chiral
  fermions on the lattice with matrix product operator norms}} (\bibinfo {year}
  {2024}),\ \Eprint {https://arxiv.org/abs/2405.10285} {arXiv:2405.10285
  [cond-mat, physics:hep-lat, physics:hep-th, physics:quant-ph]} \BibitemShut
  {NoStop}%
\bibitem [{\citenamefont {Cheluvaraja}\ and\ \citenamefont
  {Hari~Dass}(2001)}]{cheluvaraja_axial_2001}%
  \BibitemOpen
  \bibfield  {author} {\bibinfo {author} {\bibfnamefont {S.}~\bibnamefont
  {Cheluvaraja}}\ and\ \bibinfo {author} {\bibfnamefont {N.~D.}\ \bibnamefont
  {Hari~Dass}},\ }\bibfield  {title} {\bibinfo {title} {Axial anomaly and
  {{Ginsparg-Wilson}} fermions in the lattice {{Dirac}} sea picture},\ }\href
  {https://doi.org/10.1016/S0550-3213(00)00772-0} {\bibfield  {journal}
  {\bibinfo  {journal} {Nucl. Phys. B}\ }\textbf {\bibinfo {volume} {598}},\
  \bibinfo {pages} {134} (\bibinfo {year} {2001})}\BibitemShut {NoStop}%
\bibitem [{\citenamefont {Creutz}\ \emph {et~al.}(2002)\citenamefont {Creutz},
  \citenamefont {Horvath},\ and\ \citenamefont {Neuberger}}]{creutz_new_2002}%
  \BibitemOpen
  \bibfield  {author} {\bibinfo {author} {\bibfnamefont {M.}~\bibnamefont
  {Creutz}}, \bibinfo {author} {\bibfnamefont {I.}~\bibnamefont {Horvath}},\
  and\ \bibinfo {author} {\bibfnamefont {H.}~\bibnamefont {Neuberger}},\
  }\bibfield  {title} {\bibinfo {title} {A {{New}} fermion {{Hamiltonian}} for
  lattice gauge theory},\ }\href
  {https://doi.org/10.1016/S0920-5632(01)01836-9} {\bibfield  {journal}
  {\bibinfo  {journal} {Nucl. Phys. B Proc. Suppl.}\ }\textbf {\bibinfo
  {volume} {106}},\ \bibinfo {pages} {760} (\bibinfo {year}
  {2002})}\BibitemShut {NoStop}%
\bibitem [{\citenamefont {Rhodes}\ \emph {et~al.}(2024)\citenamefont {Rhodes},
  \citenamefont {Kreshchuk},\ and\ \citenamefont {Pathak}}]{Rhodes:2024zbr}%
  \BibitemOpen
  \bibfield  {author} {\bibinfo {author} {\bibfnamefont {M.~L.}\ \bibnamefont
  {Rhodes}}, \bibinfo {author} {\bibfnamefont {M.}~\bibnamefont {Kreshchuk}},\
  and\ \bibinfo {author} {\bibfnamefont {S.}~\bibnamefont {Pathak}},\
  }\bibfield  {title} {\bibinfo {title} {{Exponential Improvements in the
  Simulation of Lattice Gauge Theories Using Near-Optimal Techniques}},\ }\href
  {https://doi.org/10.1103/PRXQuantum.5.040347} {\bibfield  {journal} {\bibinfo
   {journal} {PRX Quantum}\ }\textbf {\bibinfo {volume} {5}},\ \bibinfo {pages}
  {040347} (\bibinfo {year} {2024})},\ \Eprint
  {https://arxiv.org/abs/2405.10416} {arXiv:2405.10416 [quant-ph]} \BibitemShut
  {NoStop}%
\bibitem [{\citenamefont {Trotter}(1959)}]{Trotter:1959ytf}%
  \BibitemOpen
  \bibfield  {author} {\bibinfo {author} {\bibfnamefont {H.~F.}\ \bibnamefont
  {Trotter}},\ }\bibfield  {title} {\bibinfo {title} {{On the product of
  semi-groups of operators}},\ }\href
  {https://doi.org/10.1090/s0002-9939-1959-0108732-6} {\bibfield  {journal}
  {\bibinfo  {journal} {Proc. Am. Math. Soc.}\ }\textbf {\bibinfo {volume}
  {10}},\ \bibinfo {pages} {545} (\bibinfo {year} {1959})}\BibitemShut
  {NoStop}%
\bibitem [{\citenamefont {Suzuki}(1976)}]{Suzuki:1976be}%
  \BibitemOpen
  \bibfield  {author} {\bibinfo {author} {\bibfnamefont {M.}~\bibnamefont
  {Suzuki}},\ }\bibfield  {title} {\bibinfo {title} {{Generalized Trotter's
  Formula and Systematic Approximants of Exponential Operators and Inner
  Derivations with Applications to Many Body Problems}},\ }\href
  {https://doi.org/10.1007/BF01609348} {\bibfield  {journal} {\bibinfo
  {journal} {Commun. Math. Phys.}\ }\textbf {\bibinfo {volume} {51}},\ \bibinfo
  {pages} {183} (\bibinfo {year} {1976})}\BibitemShut {NoStop}%
\bibitem [{\citenamefont {Suzuki}(1985)}]{Suzuki:1985}%
  \BibitemOpen
  \bibfield  {author} {\bibinfo {author} {\bibfnamefont {M.}~\bibnamefont
  {Suzuki}},\ }\bibfield  {title} {\bibinfo {title} {Decomposition formulas of
  exponential operators and lie exponentials with some applications to quantum
  mechanics and statistical physics},\ }\bibfield  {journal} {\bibinfo
  {journal} {Journal of Mathematical Physics}\ }\textbf {\bibinfo {volume}
  {26}},\ \href {https://doi.org/10.1063/1.526596} {10.1063/1.526596} (\bibinfo
  {year} {1985}),\ \Eprint
  {https://arxiv.org/abs/https://doi.org/10.1063/1.526596}
  {https://doi.org/10.1063/1.526596} \BibitemShut {NoStop}%
\bibitem [{\citenamefont {Childs}\ \emph {et~al.}(2021)\citenamefont {Childs},
  \citenamefont {Su}, \citenamefont {Tran}, \citenamefont {Wiebe},\ and\
  \citenamefont {Zhu}}]{PhysRevX.11.011020}%
  \BibitemOpen
  \bibfield  {author} {\bibinfo {author} {\bibfnamefont {A.~M.}\ \bibnamefont
  {Childs}}, \bibinfo {author} {\bibfnamefont {Y.}~\bibnamefont {Su}}, \bibinfo
  {author} {\bibfnamefont {M.~C.}\ \bibnamefont {Tran}}, \bibinfo {author}
  {\bibfnamefont {N.}~\bibnamefont {Wiebe}},\ and\ \bibinfo {author}
  {\bibfnamefont {S.}~\bibnamefont {Zhu}},\ }\bibfield  {title} {\bibinfo
  {title} {Theory of trotter error with commutator scaling},\ }\href
  {https://doi.org/10.1103/PhysRevX.11.011020} {\bibfield  {journal} {\bibinfo
  {journal} {Phys. Rev. X}\ }\textbf {\bibinfo {volume} {11}},\ \bibinfo
  {pages} {011020} (\bibinfo {year} {2021})}\BibitemShut {NoStop}%
\bibitem [{\citenamefont {Low}\ and\ \citenamefont
  {Chuang}(2017)}]{low2017hamiltoniansimulationuniformspectral}%
  \BibitemOpen
  \bibfield  {author} {\bibinfo {author} {\bibfnamefont {G.~H.}\ \bibnamefont
  {Low}}\ and\ \bibinfo {author} {\bibfnamefont {I.~L.}\ \bibnamefont
  {Chuang}},\ }\href {https://arxiv.org/abs/1707.05391} {\bibinfo {title}
  {Hamiltonian simulation by uniform spectral amplification}} (\bibinfo {year}
  {2017}),\ \Eprint {https://arxiv.org/abs/1707.05391} {arXiv:1707.05391
  [quant-ph]} \BibitemShut {NoStop}%
\bibitem [{\citenamefont {Low}\ and\ \citenamefont {Chuang}(2019)}]{Low_2019}%
  \BibitemOpen
  \bibfield  {author} {\bibinfo {author} {\bibfnamefont {G.~H.}\ \bibnamefont
  {Low}}\ and\ \bibinfo {author} {\bibfnamefont {I.~L.}\ \bibnamefont
  {Chuang}},\ }\bibfield  {title} {\bibinfo {title} {Hamiltonian simulation by
  qubitization},\ }\href {https://doi.org/10.22331/q-2019-07-12-163} {\bibfield
   {journal} {\bibinfo  {journal} {Quantum}\ }\textbf {\bibinfo {volume} {3}},\
  \bibinfo {pages} {163} (\bibinfo {year} {2019})}\BibitemShut {NoStop}%
\bibitem [{\citenamefont {Gilyén}\ \emph {et~al.}(2019)\citenamefont
  {Gilyén}, \citenamefont {Su}, \citenamefont {Low},\ and\ \citenamefont
  {Wiebe}}]{Gilyen:2018khw}%
  \BibitemOpen
  \bibfield  {author} {\bibinfo {author} {\bibfnamefont {A.}~\bibnamefont
  {Gilyén}}, \bibinfo {author} {\bibfnamefont {Y.}~\bibnamefont {Su}},
  \bibinfo {author} {\bibfnamefont {G.~H.}\ \bibnamefont {Low}},\ and\ \bibinfo
  {author} {\bibfnamefont {N.}~\bibnamefont {Wiebe}},\ }\bibfield  {title}
  {\bibinfo {title} {Quantum singular value transformation and beyond:
  exponential improvements for quantum matrix arithmetics},\ }in\ \href
  {https://doi.org/10.1145/3313276.3316366} {\emph {\bibinfo {booktitle}
  {Proceedings of the 51st Annual ACM SIGACT Symposium on Theory of
  Computing}}},\ \bibinfo {series and number} {STOC ’19}\ (\bibinfo
  {publisher} {ACM},\ \bibinfo {year} {2019})\ p.\ \bibinfo {pages}
  {193–204}\BibitemShut {NoStop}%
\bibitem [{\citenamefont {Berry}\ \emph {et~al.}(2015)\citenamefont {Berry},
  \citenamefont {Childs}, \citenamefont {Cleve}, \citenamefont {Kothari},\ and\
  \citenamefont {Somma}}]{Berry:2014ivo}%
  \BibitemOpen
  \bibfield  {author} {\bibinfo {author} {\bibfnamefont {D.~W.}\ \bibnamefont
  {Berry}}, \bibinfo {author} {\bibfnamefont {A.~M.}\ \bibnamefont {Childs}},
  \bibinfo {author} {\bibfnamefont {R.}~\bibnamefont {Cleve}}, \bibinfo
  {author} {\bibfnamefont {R.}~\bibnamefont {Kothari}},\ and\ \bibinfo {author}
  {\bibfnamefont {R.~D.}\ \bibnamefont {Somma}},\ }\bibfield  {title} {\bibinfo
  {title} {{Simulating Hamiltonian Dynamics with a Truncated Taylor Series}},\
  }\href {https://doi.org/10.1103/PhysRevLett.114.090502} {\bibfield  {journal}
  {\bibinfo  {journal} {Phys. Rev. Lett.}\ }\textbf {\bibinfo {volume} {114}},\
  \bibinfo {pages} {090502} (\bibinfo {year} {2015})},\ \Eprint
  {https://arxiv.org/abs/1412.4687} {arXiv:1412.4687 [quant-ph]} \BibitemShut
  {NoStop}%
\bibitem [{\citenamefont {Rajput}\ \emph {et~al.}(2022)\citenamefont {Rajput},
  \citenamefont {Roggero},\ and\ \citenamefont {Wiebe}}]{Rajput:2021khs}%
  \BibitemOpen
  \bibfield  {author} {\bibinfo {author} {\bibfnamefont {A.}~\bibnamefont
  {Rajput}}, \bibinfo {author} {\bibfnamefont {A.}~\bibnamefont {Roggero}},\
  and\ \bibinfo {author} {\bibfnamefont {N.}~\bibnamefont {Wiebe}},\ }\bibfield
   {title} {\bibinfo {title} {{Hybridized Methods for Quantum Simulation in the
  Interaction Picture}},\ }\href {https://doi.org/10.22331/q-2022-08-17-780}
  {\bibfield  {journal} {\bibinfo  {journal} {Quantum}\ }\textbf {\bibinfo
  {volume} {6}},\ \bibinfo {pages} {780} (\bibinfo {year} {2022})},\ \Eprint
  {https://arxiv.org/abs/2109.03308} {arXiv:2109.03308 [quant-ph]} \BibitemShut
  {NoStop}%
\bibitem [{\citenamefont {Kane}\ \emph {et~al.}(2024)\citenamefont {Kane},
  \citenamefont {Gomes},\ and\ \citenamefont {Kreshchuk}}]{Kane:2023jdo}%
  \BibitemOpen
  \bibfield  {author} {\bibinfo {author} {\bibfnamefont {C.~F.}\ \bibnamefont
  {Kane}}, \bibinfo {author} {\bibfnamefont {N.}~\bibnamefont {Gomes}},\ and\
  \bibinfo {author} {\bibfnamefont {M.}~\bibnamefont {Kreshchuk}},\ }\bibfield
  {title} {\bibinfo {title} {{Nearly optimal state preparation for quantum
  simulations of lattice gauge theories}},\ }\href
  {https://doi.org/10.1103/PhysRevA.110.012455} {\bibfield  {journal} {\bibinfo
   {journal} {Phys. Rev. A}\ }\textbf {\bibinfo {volume} {110}},\ \bibinfo
  {pages} {012455} (\bibinfo {year} {2024})},\ \Eprint
  {https://arxiv.org/abs/2310.13757} {arXiv:2310.13757 [quant-ph]} \BibitemShut
  {NoStop}%
\bibitem [{\citenamefont {Hariprakash}\ \emph {et~al.}(2025)\citenamefont
  {Hariprakash}, \citenamefont {Modi}, \citenamefont {Kreshchuk}, \citenamefont
  {Kane},\ and\ \citenamefont {Bauer}}]{Hariprakash:2023tla}%
  \BibitemOpen
  \bibfield  {author} {\bibinfo {author} {\bibfnamefont {S.}~\bibnamefont
  {Hariprakash}}, \bibinfo {author} {\bibfnamefont {N.~S.}\ \bibnamefont
  {Modi}}, \bibinfo {author} {\bibfnamefont {M.}~\bibnamefont {Kreshchuk}},
  \bibinfo {author} {\bibfnamefont {C.~F.}\ \bibnamefont {Kane}},\ and\
  \bibinfo {author} {\bibfnamefont {C.~W.}\ \bibnamefont {Bauer}},\ }\bibfield
  {title} {\bibinfo {title} {{Strategies for simulating the time evolution of
  Hamiltonian lattice field theories}},\ }\href
  {https://doi.org/10.1103/PhysRevA.111.022419} {\bibfield  {journal} {\bibinfo
   {journal} {Phys. Rev. A}\ }\textbf {\bibinfo {volume} {111}},\ \bibinfo
  {pages} {022419} (\bibinfo {year} {2025})},\ \Eprint
  {https://arxiv.org/abs/2312.11637} {arXiv:2312.11637 [quant-ph]} \BibitemShut
  {NoStop}%
\bibitem [{\citenamefont {Spagnoli}\ \emph {et~al.}(2026)\citenamefont
  {Spagnoli}, \citenamefont {Roggero},\ and\ \citenamefont
  {Wiebe}}]{Spagnoli:2024mib}%
  \BibitemOpen
  \bibfield  {author} {\bibinfo {author} {\bibfnamefont {L.}~\bibnamefont
  {Spagnoli}}, \bibinfo {author} {\bibfnamefont {A.}~\bibnamefont {Roggero}},\
  and\ \bibinfo {author} {\bibfnamefont {N.}~\bibnamefont {Wiebe}},\ }\bibfield
   {title} {\bibinfo {title} {{Fault-tolerant simulation of Lattice Gauge
  Theories with gauge covariant codes}},\ }\href
  {https://doi.org/10.22331/q-2026-01-16-1968} {\bibfield  {journal} {\bibinfo
  {journal} {Quantum}\ }\textbf {\bibinfo {volume} {10}},\ \bibinfo {pages}
  {1968} (\bibinfo {year} {2026})},\ \Eprint {https://arxiv.org/abs/2405.19293}
  {arXiv:2405.19293 [quant-ph]} \BibitemShut {NoStop}%
\bibitem [{\citenamefont {Kiss}\ \emph {et~al.}(2025)\citenamefont {Kiss},
  \citenamefont {Azad}, \citenamefont {Requena}, \citenamefont {Roggero},
  \citenamefont {Wakeham},\ and\ \citenamefont {Arrazola}}]{Kiss:2024sep}%
  \BibitemOpen
  \bibfield  {author} {\bibinfo {author} {\bibfnamefont {O.}~\bibnamefont
  {Kiss}}, \bibinfo {author} {\bibfnamefont {U.}~\bibnamefont {Azad}}, \bibinfo
  {author} {\bibfnamefont {B.}~\bibnamefont {Requena}}, \bibinfo {author}
  {\bibfnamefont {A.}~\bibnamefont {Roggero}}, \bibinfo {author} {\bibfnamefont
  {D.}~\bibnamefont {Wakeham}},\ and\ \bibinfo {author} {\bibfnamefont {J.~M.}\
  \bibnamefont {Arrazola}},\ }\bibfield  {title} {\bibinfo {title} {{Early
  Fault-Tolerant Quantum Algorithms in Practice: Application to Ground-State
  Energy Estimation}},\ }\href {https://doi.org/10.22331/q-2025-04-01-1682}
  {\bibfield  {journal} {\bibinfo  {journal} {Quantum}\ }\textbf {\bibinfo
  {volume} {9}},\ \bibinfo {pages} {1682} (\bibinfo {year} {2025})},\ \Eprint
  {https://arxiv.org/abs/2405.03754} {arXiv:2405.03754 [quant-ph]} \BibitemShut
  {NoStop}%
\bibitem [{\citenamefont {Kane}\ \emph {et~al.}(2025)\citenamefont {Kane},
  \citenamefont {Hariprakash}, \citenamefont {Modi}, \citenamefont
  {Kreshchuk},\ and\ \citenamefont {Bauer}}]{Kane:2024odt}%
  \BibitemOpen
  \bibfield  {author} {\bibinfo {author} {\bibfnamefont {C.~F.}\ \bibnamefont
  {Kane}}, \bibinfo {author} {\bibfnamefont {S.}~\bibnamefont {Hariprakash}},
  \bibinfo {author} {\bibfnamefont {N.~S.}\ \bibnamefont {Modi}}, \bibinfo
  {author} {\bibfnamefont {M.}~\bibnamefont {Kreshchuk}},\ and\ \bibinfo
  {author} {\bibfnamefont {C.~W.}\ \bibnamefont {Bauer}},\ }\bibfield  {title}
  {\bibinfo {title} {{Block encoding bosons by signal processing}},\ }\href
  {https://doi.org/10.22331/q-2025-05-15-1747} {\bibfield  {journal} {\bibinfo
  {journal} {Quantum}\ }\textbf {\bibinfo {volume} {9}},\ \bibinfo {pages}
  {1747} (\bibinfo {year} {2025})},\ \Eprint {https://arxiv.org/abs/2408.16824}
  {arXiv:2408.16824 [quant-ph]} \BibitemShut {NoStop}%
\bibitem [{\citenamefont {Spagnoli}\ \emph {et~al.}(2025)\citenamefont
  {Spagnoli}, \citenamefont {Lissoni},\ and\ \citenamefont
  {Roggero}}]{Spagnoli:2025xvk}%
  \BibitemOpen
  \bibfield  {author} {\bibinfo {author} {\bibfnamefont {L.}~\bibnamefont
  {Spagnoli}}, \bibinfo {author} {\bibfnamefont {C.}~\bibnamefont {Lissoni}},\
  and\ \bibinfo {author} {\bibfnamefont {A.}~\bibnamefont {Roggero}},\
  }\href@noop {} {\bibinfo {title} {{Quantum Simulation of Nuclear Dynamics in
  First Quantization}}} (\bibinfo {year} {2025}),\ \Eprint
  {https://arxiv.org/abs/2507.22814} {arXiv:2507.22814 [quant-ph]} \BibitemShut
  {NoStop}%
\bibitem [{\citenamefont {Borici}(2004)}]{borici_computational_2004}%
  \BibitemOpen
  \bibfield  {author} {\bibinfo {author} {\bibfnamefont {A.}~\bibnamefont
  {Borici}},\ }\bibfield  {title} {\bibinfo {title} {{Computational methods for
  the fermion determinant and the link between overlap and domain wall
  fermions}},\ }in\ \href@noop {} {\emph {\bibinfo {booktitle} {{3rd
  International Workshop on Numerical Analysis and Lattice QCD}}}}\ (\bibinfo
  {year} {2004})\ pp.\ \bibinfo {pages} {25--39},\ \Eprint
  {https://arxiv.org/abs/hep-lat/0402035} {arXiv:hep-lat/0402035} \BibitemShut
  {NoStop}%
\bibitem [{\citenamefont {Brower}\ \emph {et~al.}(2006)\citenamefont {Brower},
  \citenamefont {Neff},\ and\ \citenamefont {Orginos}}]{brower_mobius_2006}%
  \BibitemOpen
  \bibfield  {author} {\bibinfo {author} {\bibfnamefont {R.~C.}\ \bibnamefont
  {Brower}}, \bibinfo {author} {\bibfnamefont {H.}~\bibnamefont {Neff}},\ and\
  \bibinfo {author} {\bibfnamefont {K.}~\bibnamefont {Orginos}},\ }\bibfield
  {title} {\bibinfo {title} {Mobius fermions},\ }\href
  {https://doi.org/10.1016/j.nuclphysbps.2006.01.047} {\bibfield  {journal}
  {\bibinfo  {journal} {Nucl. Phys. B Proc. Suppl.}\ }\textbf {\bibinfo
  {volume} {153}},\ \bibinfo {pages} {191} (\bibinfo {year}
  {2006})}\BibitemShut {NoStop}%
\bibitem [{\citenamefont {Narayanan}\ and\ \citenamefont
  {Neuberger}(1993)}]{narayanan_infinitely_1993}%
  \BibitemOpen
  \bibfield  {author} {\bibinfo {author} {\bibfnamefont {R.}~\bibnamefont
  {Narayanan}}\ and\ \bibinfo {author} {\bibfnamefont {H.}~\bibnamefont
  {Neuberger}},\ }\bibfield  {title} {\bibinfo {title} {Infinitely many
  regulator fields for chiral fermions},\ }\href
  {https://doi.org/10.1016/0370-2693(93)90636-V} {\bibfield  {journal}
  {\bibinfo  {journal} {Physics Letters B}\ }\textbf {\bibinfo {volume}
  {302}},\ \bibinfo {pages} {62} (\bibinfo {year} {1993})}\BibitemShut
  {NoStop}%
\bibitem [{\citenamefont {Childs}\ and\ \citenamefont
  {Wiebe}(2012)}]{childs2012hamiltonian}%
  \BibitemOpen
  \bibfield  {author} {\bibinfo {author} {\bibfnamefont {A.~M.}\ \bibnamefont
  {Childs}}\ and\ \bibinfo {author} {\bibfnamefont {N.}~\bibnamefont {Wiebe}},\
  }\bibfield  {title} {\bibinfo {title} {{Hamiltonian Simulation Using Linear
  Combinations of Unitary Operations}},\ }\href
  {https://doi.org/10.26421/QIC12.11-12-1} {\bibfield  {journal} {\bibinfo
  {journal} {Quant. Inf. Comput.}\ }\textbf {\bibinfo {volume} {12}},\ \bibinfo
  {pages} {0901} (\bibinfo {year} {2012})},\ \Eprint
  {https://arxiv.org/abs/1202.5822} {arXiv:1202.5822 [quant-ph]} \BibitemShut
  {NoStop}%
\bibitem [{\citenamefont {Gidney}(2018)}]{Gidney_2018}%
  \BibitemOpen
  \bibfield  {author} {\bibinfo {author} {\bibfnamefont {C.}~\bibnamefont
  {Gidney}},\ }\bibfield  {title} {\bibinfo {title} {Halving the cost of
  quantum addition},\ }\href {https://doi.org/10.22331/q-2018-06-18-74}
  {\bibfield  {journal} {\bibinfo  {journal} {Quantum}\ }\textbf {\bibinfo
  {volume} {2}},\ \bibinfo {pages} {74} (\bibinfo {year} {2018})}\BibitemShut
  {NoStop}%
\bibitem [{\citenamefont {Plesch}\ and\ \citenamefont
  {Brukner}(2011)}]{Plesch_2011}%
  \BibitemOpen
  \bibfield  {author} {\bibinfo {author} {\bibfnamefont {M.}~\bibnamefont
  {Plesch}}\ and\ \bibinfo {author} {\bibfnamefont {{\v{C}}.}~\bibnamefont
  {Brukner}},\ }\bibfield  {title} {\bibinfo {title} {Quantum-state preparation
  with universal gate decompositions},\ }\bibfield  {journal} {\bibinfo
  {journal} {Physical Review A}\ }\textbf {\bibinfo {volume} {83}},\ \href
  {https://doi.org/10.1103/physreva.83.032302} {10.1103/physreva.83.032302}
  (\bibinfo {year} {2011})\BibitemShut {NoStop}%
\bibitem [{\citenamefont {Babbush}\ \emph {et~al.}(2018)\citenamefont
  {Babbush}, \citenamefont {Gidney}, \citenamefont {Berry}, \citenamefont
  {Wiebe}, \citenamefont {McClean}, \citenamefont {Paler}, \citenamefont
  {Fowler},\ and\ \citenamefont {Neven}}]{PhysRevX.8.041015}%
  \BibitemOpen
  \bibfield  {author} {\bibinfo {author} {\bibfnamefont {R.}~\bibnamefont
  {Babbush}}, \bibinfo {author} {\bibfnamefont {C.}~\bibnamefont {Gidney}},
  \bibinfo {author} {\bibfnamefont {D.~W.}\ \bibnamefont {Berry}}, \bibinfo
  {author} {\bibfnamefont {N.}~\bibnamefont {Wiebe}}, \bibinfo {author}
  {\bibfnamefont {J.}~\bibnamefont {McClean}}, \bibinfo {author} {\bibfnamefont
  {A.}~\bibnamefont {Paler}}, \bibinfo {author} {\bibfnamefont
  {A.}~\bibnamefont {Fowler}},\ and\ \bibinfo {author} {\bibfnamefont
  {H.}~\bibnamefont {Neven}},\ }\bibfield  {title} {\bibinfo {title} {Encoding
  electronic spectra in quantum circuits with linear t complexity},\ }\href
  {https://doi.org/10.1103/PhysRevX.8.041015} {\bibfield  {journal} {\bibinfo
  {journal} {Phys. Rev. X}\ }\textbf {\bibinfo {volume} {8}},\ \bibinfo {pages}
  {041015} (\bibinfo {year} {2018})}\BibitemShut {NoStop}%
\bibitem [{\citenamefont {Luscher}(2007{\natexlab{a}})}]{Luscher:2007se}%
  \BibitemOpen
  \bibfield  {author} {\bibinfo {author} {\bibfnamefont {M.}~\bibnamefont
  {Luscher}},\ }\bibfield  {title} {\bibinfo {title} {{Local coherence and
  deflation of the low quark modes in lattice QCD}},\ }\href
  {https://doi.org/10.1088/1126-6708/2007/07/081} {\bibfield  {journal}
  {\bibinfo  {journal} {JHEP}\ }\textbf {\bibinfo {volume} {07}},\ \bibinfo
  {pages} {081}},\ \Eprint {https://arxiv.org/abs/0706.2298} {arXiv:0706.2298
  [hep-lat]} \BibitemShut {NoStop}%
\bibitem [{\citenamefont {Luscher}(2007{\natexlab{b}})}]{Luscher:2007es}%
  \BibitemOpen
  \bibfield  {author} {\bibinfo {author} {\bibfnamefont {M.}~\bibnamefont
  {Luscher}},\ }\bibfield  {title} {\bibinfo {title} {{Deflation acceleration
  of lattice QCD simulations}},\ }\href
  {https://doi.org/10.1088/1126-6708/2007/12/011} {\bibfield  {journal}
  {\bibinfo  {journal} {JHEP}\ }\textbf {\bibinfo {volume} {12}},\ \bibinfo
  {pages} {011}},\ \Eprint {https://arxiv.org/abs/0710.5417} {arXiv:0710.5417
  [hep-lat]} \BibitemShut {NoStop}%
\bibitem [{\citenamefont {Somma}\ and\ \citenamefont
  {Boixo}(2013)}]{doi:10.1137/120871997}%
  \BibitemOpen
  \bibfield  {author} {\bibinfo {author} {\bibfnamefont {R.~D.}\ \bibnamefont
  {Somma}}\ and\ \bibinfo {author} {\bibfnamefont {S.}~\bibnamefont {Boixo}},\
  }\bibfield  {title} {\bibinfo {title} {Spectral gap amplification},\ }\href
  {https://doi.org/10.1137/120871997} {\bibfield  {journal} {\bibinfo
  {journal} {SIAM Journal on Computing}\ }\textbf {\bibinfo {volume} {42}},\
  \bibinfo {pages} {593} (\bibinfo {year} {2013})}\BibitemShut {NoStop}%
\bibitem [{\citenamefont {Morales}\ \emph {et~al.}(2026)\citenamefont
  {Morales}, \citenamefont {Pira}, \citenamefont {Schleich}, \citenamefont
  {Koor}, \citenamefont {Costa}, \citenamefont {An}, \citenamefont
  {Aspuru-Guzik}, \citenamefont {Lin}, \citenamefont {Rebentrost},\ and\
  \citenamefont {Berry}}]{x6gh-d8gh}%
  \BibitemOpen
  \bibfield  {author} {\bibinfo {author} {\bibfnamefont {M.~E.~S.}\
  \bibnamefont {Morales}}, \bibinfo {author} {\bibfnamefont {L.}~\bibnamefont
  {Pira}}, \bibinfo {author} {\bibfnamefont {P.}~\bibnamefont {Schleich}},
  \bibinfo {author} {\bibfnamefont {K.}~\bibnamefont {Koor}}, \bibinfo {author}
  {\bibfnamefont {P.~C.~S.}\ \bibnamefont {Costa}}, \bibinfo {author}
  {\bibfnamefont {D.}~\bibnamefont {An}}, \bibinfo {author} {\bibfnamefont
  {A.}~\bibnamefont {Aspuru-Guzik}}, \bibinfo {author} {\bibfnamefont
  {L.}~\bibnamefont {Lin}}, \bibinfo {author} {\bibfnamefont {P.}~\bibnamefont
  {Rebentrost}},\ and\ \bibinfo {author} {\bibfnamefont {D.~W.}\ \bibnamefont
  {Berry}},\ }\bibfield  {title} {\bibinfo {title} {Quantum linear system
  solvers: A survey of algorithms and applications},\ }\href
  {https://doi.org/10.1103/x6gh-d8gh} {\bibfield  {journal} {\bibinfo
  {journal} {Rev. Mod. Phys.}\ }\textbf {\bibinfo {volume} {98}},\ \bibinfo
  {pages} {025005} (\bibinfo {year} {2026})}\BibitemShut {NoStop}%
\bibitem [{\citenamefont {Costa}\ \emph {et~al.}(2022)\citenamefont {Costa},
  \citenamefont {An}, \citenamefont {Sanders}, \citenamefont {Su},
  \citenamefont {Babbush},\ and\ \citenamefont {Berry}}]{PRXQuantum.3.040303}%
  \BibitemOpen
  \bibfield  {author} {\bibinfo {author} {\bibfnamefont {P.~C.}\ \bibnamefont
  {Costa}}, \bibinfo {author} {\bibfnamefont {D.}~\bibnamefont {An}}, \bibinfo
  {author} {\bibfnamefont {Y.~R.}\ \bibnamefont {Sanders}}, \bibinfo {author}
  {\bibfnamefont {Y.}~\bibnamefont {Su}}, \bibinfo {author} {\bibfnamefont
  {R.}~\bibnamefont {Babbush}},\ and\ \bibinfo {author} {\bibfnamefont {D.~W.}\
  \bibnamefont {Berry}},\ }\bibfield  {title} {\bibinfo {title} {Optimal
  scaling quantum linear-systems solver via discrete adiabatic theorem},\
  }\href {https://doi.org/10.1103/PRXQuantum.3.040303} {\bibfield  {journal}
  {\bibinfo  {journal} {PRX Quantum}\ }\textbf {\bibinfo {volume} {3}},\
  \bibinfo {pages} {040303} (\bibinfo {year} {2022})}\BibitemShut {NoStop}%
\bibitem [{\citenamefont {Mori}\ \emph {et~al.}(2026)\citenamefont {Mori},
  \citenamefont {Kikuchi}, \citenamefont {Benedetti},\ and\ \citenamefont
  {Rosenkranz}}]{10.1088/2058-9565/ae89e0}%
  \BibitemOpen
  \bibfield  {author} {\bibinfo {author} {\bibfnamefont {H.}~\bibnamefont
  {Mori}}, \bibinfo {author} {\bibfnamefont {Y.}~\bibnamefont {Kikuchi}},
  \bibinfo {author} {\bibfnamefont {M.}~\bibnamefont {Benedetti}},\ and\
  \bibinfo {author} {\bibfnamefont {M.}~\bibnamefont {Rosenkranz}},\ }\bibfield
   {title} {\bibinfo {title} {Sparsity-dependent complexity lower bound of
  quantum linear system solvers},\ }\href
  {http://iopscience.iop.org/article/10.1088/2058-9565/ae89e0} {\bibfield
  {journal} {\bibinfo  {journal} {Quantum Science and Technology}\ } (\bibinfo
  {year} {2026})}\BibitemShut {NoStop}%
\bibitem [{\citenamefont {Haah}\ \emph {et~al.}(2023)\citenamefont {Haah},
  \citenamefont {Hastings}, \citenamefont {Kothari},\ and\ \citenamefont
  {Low}}]{Haah:2018ekc}%
  \BibitemOpen
  \bibfield  {author} {\bibinfo {author} {\bibfnamefont {J.}~\bibnamefont
  {Haah}}, \bibinfo {author} {\bibfnamefont {M.~B.}\ \bibnamefont {Hastings}},
  \bibinfo {author} {\bibfnamefont {R.}~\bibnamefont {Kothari}},\ and\ \bibinfo
  {author} {\bibfnamefont {G.~H.}\ \bibnamefont {Low}},\ }\bibfield  {title}
  {\bibinfo {title} {{Quantum Algorithm for Simulating Real Time Evolution of
  Lattice Hamiltonians}},\ }\href {https://doi.org/10.1137/18m1231511}
  {\bibfield  {journal} {\bibinfo  {journal} {SIAM J. Comput.}\ }\textbf
  {\bibinfo {volume} {52}},\ \bibinfo {pages} {FOCS18} (\bibinfo {year}
  {2023})},\ \Eprint {https://arxiv.org/abs/1801.03922} {arXiv:1801.03922
  [quant-ph]} \BibitemShut {NoStop}%
\bibitem [{\citenamefont {Chiu}\ \emph {et~al.}(2002)\citenamefont {Chiu},
  \citenamefont {Hsieh}, \citenamefont {Huang},\ and\ \citenamefont
  {Huang}}]{Chiu:2002eh}%
  \BibitemOpen
  \bibfield  {author} {\bibinfo {author} {\bibfnamefont {T.-W.}\ \bibnamefont
  {Chiu}}, \bibinfo {author} {\bibfnamefont {T.-H.}\ \bibnamefont {Hsieh}},
  \bibinfo {author} {\bibfnamefont {C.-H.}\ \bibnamefont {Huang}},\ and\
  \bibinfo {author} {\bibfnamefont {T.-R.}\ \bibnamefont {Huang}},\ }\bibfield
  {title} {\bibinfo {title} {{A Note on the Zolotarev optimal rational
  approximation for the overlap Dirac operator}},\ }\href
  {https://doi.org/10.1103/PhysRevD.66.114502} {\bibfield  {journal} {\bibinfo
  {journal} {Phys. Rev. D}\ }\textbf {\bibinfo {volume} {66}},\ \bibinfo
  {pages} {114502} (\bibinfo {year} {2002})},\ \Eprint
  {https://arxiv.org/abs/hep-lat/0206007} {arXiv:hep-lat/0206007} \BibitemShut
  {NoStop}%
\bibitem [{\citenamefont {van~den Eshof}\ \emph {et~al.}(2002)\citenamefont
  {van~den Eshof}, \citenamefont {Frommer}, \citenamefont {Lippert},
  \citenamefont {Schilling},\ and\ \citenamefont {van~der
  Vorst}}]{vandenEshof:2002ms}%
  \BibitemOpen
  \bibfield  {author} {\bibinfo {author} {\bibfnamefont {J.}~\bibnamefont
  {van~den Eshof}}, \bibinfo {author} {\bibfnamefont {A.}~\bibnamefont
  {Frommer}}, \bibinfo {author} {\bibfnamefont {T.}~\bibnamefont {Lippert}},
  \bibinfo {author} {\bibfnamefont {K.}~\bibnamefont {Schilling}},\ and\
  \bibinfo {author} {\bibfnamefont {H.~A.}\ \bibnamefont {van~der Vorst}},\
  }\bibfield  {title} {\bibinfo {title} {{Numerical methods for the QCD overlap
  operator. I. Sign function and error bounds}},\ }\href
  {https://doi.org/10.1016/S0010-4655(02)00455-1} {\bibfield  {journal}
  {\bibinfo  {journal} {Comput. Phys. Commun.}\ }\textbf {\bibinfo {volume}
  {146}},\ \bibinfo {pages} {203} (\bibinfo {year} {2002})},\ \Eprint
  {https://arxiv.org/abs/hep-lat/0202025} {arXiv:hep-lat/0202025} \BibitemShut
  {NoStop}%
\bibitem [{\citenamefont {Lamm}\ \emph {et~al.}(2019)\citenamefont {Lamm},
  \citenamefont {Lawrence},\ and\ \citenamefont {Yamauchi}}]{Lamm:2019bik}%
  \BibitemOpen
  \bibfield  {author} {\bibinfo {author} {\bibfnamefont {H.}~\bibnamefont
  {Lamm}}, \bibinfo {author} {\bibfnamefont {S.}~\bibnamefont {Lawrence}},\
  and\ \bibinfo {author} {\bibfnamefont {Y.}~\bibnamefont {Yamauchi}},\
  }\bibfield  {title} {\bibinfo {title} {{General Methods for Digital Quantum
  Simulation of Gauge Theories}},\ }\href
  {https://doi.org/10.1103/PhysRevD.100.034518} {\bibfield  {journal} {\bibinfo
   {journal} {Phys. Rev. D}\ }\textbf {\bibinfo {volume} {100}},\ \bibinfo
  {pages} {034518} (\bibinfo {year} {2019})},\ \Eprint
  {https://arxiv.org/abs/1903.08807} {arXiv:1903.08807 [hep-lat]} \BibitemShut
  {NoStop}%
\bibitem [{\citenamefont {Bravyi}\ and\ \citenamefont
  {Kitaev}(2002)}]{Bravyi:2000vfj}%
  \BibitemOpen
  \bibfield  {author} {\bibinfo {author} {\bibfnamefont {S.~B.}\ \bibnamefont
  {Bravyi}}\ and\ \bibinfo {author} {\bibfnamefont {A.~Y.}\ \bibnamefont
  {Kitaev}},\ }\bibfield  {title} {\bibinfo {title} {{Fermionic Quantum
  Computation}},\ }\href {https://doi.org/10.1006/aphy.2002.6254} {\bibfield
  {journal} {\bibinfo  {journal} {Annals Phys.}\ }\textbf {\bibinfo {volume}
  {298}},\ \bibinfo {pages} {210} (\bibinfo {year} {2002})},\ \Eprint
  {https://arxiv.org/abs/quant-ph/0003137} {arXiv:quant-ph/0003137}
  \BibitemShut {NoStop}%
\bibitem [{\citenamefont {Havl{\'\i}{\v{c}}ek}\ \emph
  {et~al.}(2017)\citenamefont {Havl{\'\i}{\v{c}}ek}, \citenamefont {Troyer},\
  and\ \citenamefont {Whitfield}}]{Havlicek:2017vcq}%
  \BibitemOpen
  \bibfield  {author} {\bibinfo {author} {\bibfnamefont {V.}~\bibnamefont
  {Havl{\'\i}{\v{c}}ek}}, \bibinfo {author} {\bibfnamefont {M.}~\bibnamefont
  {Troyer}},\ and\ \bibinfo {author} {\bibfnamefont {J.~D.}\ \bibnamefont
  {Whitfield}},\ }\bibfield  {title} {\bibinfo {title} {{Operator locality in
  the quantum simulation of fermionic models}},\ }\href
  {https://doi.org/10.1103/PhysRevA.95.032332} {\bibfield  {journal} {\bibinfo
  {journal} {Phys. Rev. A}\ }\textbf {\bibinfo {volume} {95}},\ \bibinfo
  {pages} {032332} (\bibinfo {year} {2017})},\ \Eprint
  {https://arxiv.org/abs/1701.07072} {arXiv:1701.07072 [quant-ph]} \BibitemShut
  {NoStop}%
\bibitem [{\citenamefont {Grover}(2000)}]{PhysRevLett.85.1334}%
  \BibitemOpen
  \bibfield  {author} {\bibinfo {author} {\bibfnamefont {L.~K.}\ \bibnamefont
  {Grover}},\ }\bibfield  {title} {\bibinfo {title} {Synthesis of quantum
  superpositions by quantum computation},\ }\href
  {https://doi.org/10.1103/PhysRevLett.85.1334} {\bibfield  {journal} {\bibinfo
   {journal} {Phys. Rev. Lett.}\ }\textbf {\bibinfo {volume} {85}},\ \bibinfo
  {pages} {1334} (\bibinfo {year} {2000})}\BibitemShut {NoStop}%
\bibitem [{\citenamefont {H\o{}yer}(2000)}]{PhysRevA.62.052304}%
  \BibitemOpen
  \bibfield  {author} {\bibinfo {author} {\bibfnamefont {P.}~\bibnamefont
  {H\o{}yer}},\ }\bibfield  {title} {\bibinfo {title} {Arbitrary phases in
  quantum amplitude amplification},\ }\href
  {https://doi.org/10.1103/PhysRevA.62.052304} {\bibfield  {journal} {\bibinfo
  {journal} {Phys. Rev. A}\ }\textbf {\bibinfo {volume} {62}},\ \bibinfo
  {pages} {052304} (\bibinfo {year} {2000})}\BibitemShut {NoStop}%
\bibitem [{\citenamefont {Lee}\ \emph {et~al.}(2021)\citenamefont {Lee},
  \citenamefont {Berry}, \citenamefont {Gidney}, \citenamefont {Huggins},
  \citenamefont {McClean}, \citenamefont {Wiebe},\ and\ \citenamefont
  {Babbush}}]{PRXQuantum.2.030305}%
  \BibitemOpen
  \bibfield  {author} {\bibinfo {author} {\bibfnamefont {J.}~\bibnamefont
  {Lee}}, \bibinfo {author} {\bibfnamefont {D.~W.}\ \bibnamefont {Berry}},
  \bibinfo {author} {\bibfnamefont {C.}~\bibnamefont {Gidney}}, \bibinfo
  {author} {\bibfnamefont {W.~J.}\ \bibnamefont {Huggins}}, \bibinfo {author}
  {\bibfnamefont {J.~R.}\ \bibnamefont {McClean}}, \bibinfo {author}
  {\bibfnamefont {N.}~\bibnamefont {Wiebe}},\ and\ \bibinfo {author}
  {\bibfnamefont {R.}~\bibnamefont {Babbush}},\ }\bibfield  {title} {\bibinfo
  {title} {Even more efficient quantum computations of chemistry through tensor
  hypercontraction},\ }\href {https://doi.org/10.1103/PRXQuantum.2.030305}
  {\bibfield  {journal} {\bibinfo  {journal} {PRX Quantum}\ }\textbf {\bibinfo
  {volume} {2}},\ \bibinfo {pages} {030305} (\bibinfo {year}
  {2021})}\BibitemShut {NoStop}%
\end{thebibliography}%

\onecolumngrid
\begin{center}
  \larger[2]\textsc{Supplementary Material}
\end{center}
\vspace{3em}
\twocolumngrid

\section{Block encoding the quadratic fermionic Hamiltonian}

In this section we will construct the block-encoding for a generic quadratic fermionic Hamiltonian, given an oracle that can implement the single-particle Hamiltonian $\mathbf{h}$. In general, we want to find the unitary $W$ such that
\begin{equation}
\label{eq:supp_quadratic_H}
    \langle 0^{\otimes b}\lvert W \rvert 0^{\otimes b}\rangle = \frac{H}{\Lambda}\;,
\end{equation}
where $H$ is defined in \cref{eq:snd_comp}. For ease of notation we will use $\mathbf{h}_{ab}=h_{ab} U_{ab}$ to denote the components of the single particle Hamiltonian. For the general second quantized Hamiltonian $H$ to be hermitian we will need
\begin{equation}
\label{eq:herm_cond}
h_{ab} U_{ab}=\mathbf{h}_{ab}=\mathbf{h}_{ba}^\dagger=h_{ba}^* U_{ba}^\dagger\;.
\end{equation}

Our construction starts by promoting $\mathbf{h}_{ab}$ from an operator acting only on $\mathcal{H}_G$ to an operator
\begin{equation}
\hat{\mathbf{h}}_{ab} =(h_{ab}\rvert a\rangle\langle b\lvert)\otimes U_{ab}\equiv \hat{h}_{ab}\otimes U_{ab}\;,
\end{equation}
acting also on register with $q\geq\lceil\log_2(\mathcal{Q})\rceil$ qubits. This makes it clear that we are promoting $h_{ab}$ from a $\mathcal{Q}\times\mathcal{Q}$ matrix of numbers to an operator $\hat{h}_{ab}$ acting on register with $q$ qubits. The full single particle Hamiltonian in this space then reads
\begin{equation}
\label{eq:hbf_def_supp}
\hat{\mathbf{h}}=\sum_{a,b=0}^{\mathcal{Q}-1} \hat{h}_{ab}\otimes U_{ab}=\sum_{a,b=0}^{\mathcal{Q}-1} h_{ab}\rvert a\rangle\langle b\lvert\otimes U_{ab}\;.
\end{equation}
From \cref{eq:herm_cond}, we see that $\hat{\mathbf{h}}$ is hermitian. To fix conventions, we take the size of the auxiliary qubit register as $q=n_I+n_S+n_X$ where
\begin{equation}
\begin{split}
\label{eq:qdef}
n_I &= \left\lceil\log_2(\mathcal{N}_I)\right\rceil, \; n_S = \frac{d+1}{2},\;
n_X = d\left\lceil\log_2(N)\right\rceil.
\end{split}
\end{equation}
With this convention the first $n_I$ qubits encode internal quantum numbers, the next $n_S$ the spin and the last $n_X$ the spatial coordinates. For practical purposes it is convenient to choose the number of sites per spatial dimension $N=2^n$ so that $n_X=dn$ but our construction can be extended to the case where $N$ is not a power of $2$.

A state of the auxiliary register will then have the following decomposition
\begin{equation}
    \lvert a\rangle = \lvert \cdot \rangle_{I} \lvert \cdot\rangle_S \lvert \cdot\rangle_{x_1} \cdots \lvert \cdot\rangle_{x_d}\;,
    \label{eq:aux_reg_convention_supp}
\end{equation}
with $\lvert \cdot\rangle_{x_k}$ being the register encoding the spatial coordinate in the $k$-th direction.

\begin{theorem}[Encoding of Fermionic Hamiltonian]
    Let $\hat{\mathbf{h}}$ as in \cref{eq:hbf_def_supp}. It is possible to implement two unitary operators $\mathcal{U}_L$ and $\mathcal{U}_R$ such that 
    \begin{equation}
    \langle0^{\otimes(q+2)}\lvert \mathcal{U}_L\hat{\mathbf{h}}\, \mathcal{U}_R\rvert0^{\otimes(q+2)}\rangle=\frac{H}{2^q}
    \end{equation}
    with $q=n_I+n_S+n_X\geq\log_2(\mathcal{Q})$ as in \cref{eq:qdef} with $O(\mathcal{Q})$ gates and $O(\log(\mathcal{Q}))$ auxiliary qubits.
\end{theorem}
\begin{proof}
    Consider the unitaries $\mathcal{V}_{L/R}$ acting on a $2$-qubit register $A$, the $q$-qubit register $B$ used to promote $\mathbf{h}$ to $\hat{\mathbf{h}}$, and the $\mathcal{Q}$-qubit register encoding the fermionic Fock space as follows
    \begin{equation}
    \label{eq:ulr_def}
        \mathcal{V}_{L/R} = \sum_{i=0}^{3}\sum_{a=0}^{\mathcal{Q}-1} \lvert i\rangle \langle i \rvert_A \lvert a\rangle \langle a \rvert_B \otimes V_{i,a}^{L/R}\;,
    \end{equation}
    with unitaries $V_{i,a}^{L/R}$ acting on the last $\mathcal{Q}$-qubit register. By denoting with $\rvert\phi_+\rangle=\rvert+^{\otimes(q+2)}\rangle$ the state where every qubit in the $A$ and $B$ register is in the $\lvert + \rangle$ state, we get:
    \begin{equation}
        \langle \phi_+ \rvert \mathcal{V}_L \hat{\mathbf{h}}\, \mathcal{V}_R \lvert \phi_+ \rangle = \frac{1}{2^{q+2}} \sum_{i,j=0}^{3} \sum_{a,b=0}^{\mathcal{Q}-1}  \hat{\mathbf{h}}_{a,b} V_{i,a}^LV^R_{j,b}\;.
    \end{equation}
    With the Jordan-Wigner mapping, the fermionic operators are represented as $\Psi_a = \mathcal{P}_a(X_a+iY_a)/2$ where $X_a$ and $Y_a$ are the corresponding Pauli matrices acting on qubit $a$ and we introduced a diagonal operator $\mathcal{P}_a=\prod_{c=0}^{a-1}Z_c$ where we chose a fixed ordering of the fermionic modes.
    
    If we consider the following choice:
    \begin{equation}
    \begin{split}
        V_{0,a}^L &= V_{0,a}^R = \mathcal{P}_a X_a \quad\quad\quad V_{1,a}^L = V_{1,a}^R = \mathcal{P}_a Y_a \\
        V_{2,a}^L &= V_{3,a}^R = \mathcal{P}_a X_a \quad\quad\quad V_{3,a}^L =-V_{2,a}^R = -i\mathcal{P}_a Y_a \\
    \end{split}
    \end{equation}
    then we have that the sums over $i$ and $j$ labels gives
    \begin{equation}
    \begin{split}
    \sum_{i,j=0}^{3}V_{i,a}^LV^R_{j,b}=& \mathcal{P}_a\mathcal{P}_b(X_aX_b+Y_aY_b)\\
    &+i\mathcal{P}_a\mathcal{P}_b(X_aY_b-Y_aX_b)=4\Psi^\dagger_a \Psi_b\;.
    \end{split}
    \end{equation}
    But then, if we define the unitaries
    \begin{equation}
    \mathcal{U}_L = H^{\otimes (q+2)}\mathcal{V}_L\quad\text{and}\quad\mathcal{U}_R = \mathcal{V}_RH^{\otimes (q+2)}\;,
    \end{equation}
    where $H$ is the Hadamard gate, we finally find
    \begin{equation}
         \langle 0^{\otimes(q+2)} \rvert \mathcal{U}_L \hat{\mathbf{h}}\, \mathcal{U}_R \lvert 0^{\otimes(q+2)} \rangle = \frac{H}{2^q}\;.
    \end{equation}
    Implementing the unitaries $\mathcal{V}_L$ and $\mathcal{V}_R$ as in \cref{eq:ulr_def} can be performed using a unary iteration scheme~\cite{PhysRevX.8.041015} with $O(\log(\mathcal{Q}))$ ancilla qubits, $O(\mathcal{Q})$ $T$ gates and the $4\mathcal{Q}$ unitaries $V_{i,a}^{L/R}$ controlled by a single qubit. A naive implementation of the latter will require $O(\mathcal{Q}^2)$ Clifford operations but adding an accumulator qubit to the unary iterator this can be brought down to only $O(\mathcal{Q})$ instead~\cite{PhysRevX.8.041015}.

\end{proof}

Considering again \cref{eq:supp_quadratic_H}, we proved that it is possible to implement the block-encoding unitary $W=\mathcal{U}_L \hat{\mathbf{h}}\, \mathcal{U}_R$, with the subnormalization $\Lambda = 2^q$, and the size of the auxiliary register $b$ is exactly $b=q$.

In situations where $2^q>\mathcal{Q}$ it is possible to improve the normalization of the block encoded Hamiltonian from $2^q$ to $\mathcal{Q}$ by replacing the state $\lvert +^{\otimes(q+2)} \rangle$ with the state $\frac{1}{\sqrt{4\mathcal{Q}}}\lvert ++ \rangle_A \sum_{a=0}^{\mathcal{Q}-1} \lvert a\rangle_B$. This can be prepared with $O(\log(\mathcal{Q}/\epsilon))$ gates at the price of introducing an error $\epsilon$~\cite{PhysRevLett.85.1334,PhysRevA.62.052304,PhysRevX.8.041015,PRXQuantum.2.030305}. In this work we consider the simpler case where we keep the block encoding exact at the price of a possible factor of $\approx2$ in the normalization.

\section{GW violation}

In this section, a bound on the Ginsparg--Wilson (GW)
relation violation is derived, when considering an approximated overlap single-particle Hamiltonian.

Let $V = \gamma^0\varepsilon(h_w)$, with $\varepsilon(h_w) = h_w/|h_w|$ the exact sign operator, and let $\hat{V} = \gamma^0 E_M$, with $E_M$ the QSP approximation of the sign function. Those operators satisfy \cref{eq:BE_polynomial_approx_sign}, that we write here as
\begin{equation}
    \|\hat{V} - V\| \leq \epsilon_e\;.
\end{equation}

We define the approximated chirality operator for overlap fermions, as well as the approximated single-particle Hamiltonian (from \cref{eq:hov,eq:modified_gamma_5}) as:
\begin{equation}
  \hat{\gamma}_5 = \tfrac{1}{2}\gamma_5\!\left(1 - \hat{V}\right),
  \qquad
  \hat{h} = \gamma^0\!\left(1 + \hat{V}\right).
  \label{eq:defs}
\end{equation}
Let us denote with $\Gamma_i = \gamma_0\gamma_i$, $\Gamma_0 = \gamma_0$ and $\Gamma_5 = \gamma_0\gamma_5$, so that $\Gamma_5\Gamma_i\Gamma_5^\dagger = -\Gamma_i$ and $\Gamma_5\Gamma_0\Gamma_5^\dagger = -\Gamma_0$. To prove this, it is enough to use the anti-commutation of gamma matrices, and the fact that $\gamma_0^2 = \mathbb{1}$. Moreover, recall the Wilson-Dirac single-particle Hamiltonian as
\begin{equation}
    \hw = i \sum_{k=1}^d \Gamma^k \delta_k + \Gamma_0 \left(m-r\frac{\Delta}{2}\right)\;.
\end{equation}
Since every term of $\hw$ contains either $\Gamma_i$ or $\Gamma_0$ but not $\Gamma_5$, it is easy to see that 
\begin{equation}
    \Gamma_5 \hw \Gamma_5^\dagger = -\hw\;.
    \label{eq:hw_conjugation_Gamma5}
\end{equation}

Then, we will use the following identities:
\begin{enumerate}
    \item $\gamma_0 V \gamma_0 = V^\dagger$, $\gamma_0 \hat{V} \gamma_0 = \hat{V}^\dagger$
    \item $\gamma_5 V \gamma_5 = V^\dagger$, $\gamma_5 \hat{V} \gamma_5 = \hat{V}^\dagger$
\end{enumerate}
We can show the conjugation relations by $\gamma_0$ as
\begin{equation}
    \gamma_0 V \gamma_0 = \gamma_0 \gamma_0 \varepsilon(\hw) \gamma_0 = \varepsilon(\hw) \gamma_0 = V^\dagger
\end{equation}
and the same is true for the approximated $\hat{V} = \gamma_0 E_M$, knowing that $E_M$ is hermitian. This is true because $E_M = \mathcal{P}_M(\hwopG/\alpha\Lambda_G)$ with $\mathcal{P}_M$ a real polynomial of definite parity, and a real polynomial of a Hermitian operator is Hermitian. As for the relations with the conjugation by $\gamma_5$, we can show it as:
\begin{equation}
\begin{split}
    \gamma_5 V^\dagger \gamma_5 &= \gamma_5 \varepsilon(\hw)\gamma_0 \gamma_5 \\
    &= \gamma_0 \gamma_0 \gamma_5 \varepsilon(\hw)\gamma_0 \gamma_5 \\
    &= \gamma_0 \Gamma_5 \varepsilon(\hw) \Gamma_5 \\
    &= -\gamma_0 \Gamma_5 \varepsilon(\hw) \Gamma_5^\dagger \\
    &= -\gamma_0 \varepsilon( \Gamma_5\hw\Gamma_5^\dagger) \\
    &= - \gamma_0 \varepsilon(-\hw) \\
    &= \gamma_0\varepsilon(\hw) = V
\end{split}
\end{equation}
where we used \cref{eq:hw_conjugation_Gamma5}, and the fact that for every unitary $U$, and function $f(x)$ that admits a taylor expansion, we have that
\begin{equation}
    Uf(x)U^\dagger = f(UxU^\dagger)\;.
\end{equation}
The same proof can be carried out for the approximated $\hat{V}$ in the same way, considering $\hat{V}=\gamma_0 E_M$, where $E_M=E_M(\hw)$ is a polynomial (smooth function) of $\hw$.

Now that we have all the identities needed, we can go to the actual bound of the \ac{GW} relation violation. Using Identity~2 we have $\gamma_5(1-\hat{V}) = (1-\hat{V}^\dagger)\gamma_5$, and
we compute the two orderings of $\hat{\gamma}_5\hat{h}$:
\begin{align}
  \hat{\gamma}_5\,\hat{h}
    &= \tfrac{1}{2}(1-\hat{V}^\dagger)\gamma_5\gamma^0(1+\hat{V})
    \notag \\
    &= -\tfrac{1}{2}(1-\hat{V}^\dagger)\gamma^0(1+\hat{V}^\dagger)\gamma_5,
  \label{eq:lhs}
\end{align}
Similarly,
\begin{equation}
  \hat{h}\,\hat{\gamma}_5
  = \tfrac{1}{2}\,\gamma^0(1+\hat{V})(1-\hat{V}^\dagger)\gamma_5.
  \label{eq:rhs}
\end{equation}
The commutator is therefore
\begin{equation}
  [\hat{\gamma}_5,\,\hat{h}]
  = -\tfrac{1}{2}
    \Bigl[
      (1-\hat{V}^\dagger)\gamma^0(1+\hat{V}^\dagger)
      +
      \gamma^0(1+\hat{V})(1-\hat{V}^\dagger)
    \Bigr]\gamma_5.
  \label{eq:comm}
\end{equation}
We expand each bracket using Identity~1 to eliminate $\hat{V}^\dagger\gamma^0$:
\begin{align}
  (1-\hat{V}^\dagger)\gamma^0(1+\hat{V}^\dagger)
    &= \gamma^0(\hat{V}^\dagger-\hat{V}) + \gamma^0(1-\hat{V}\hat{V}^\dagger), \\
  \gamma^0(1+\hat{V})(1-\hat{V}^\dagger)
    &= \gamma^0(\hat{V}-\hat{V}^\dagger) + \gamma^0(1-\hat{V}\hat{V}^\dagger).
\end{align}
Summing these two expressions, the skew-Hermitian parts $\gamma^0(\hat{V}^\dagger - \hat{V})$
cancel, giving $2\gamma^0(1 - \hat{V}\hat{V}^\dagger)$. Substituting back into~\eqref{eq:comm}:
\begin{equation}
  [\hat{\gamma}_5,\,\hat{h}] = -\gamma^0(1-\hat{V}\hat{V}^\dagger)\,\gamma_5.
  \label{eq:comm_final}
\end{equation}
In the exact case $\hat{V} = V$, unitarity gives $VV^\dagger = 1$ and the commutator
vanishes, recovering the exact chiral symmetry. However, since $\hat{V} = \gamma^0 E_M$, we have $\hat{V}\hat{V}^\dagger = \gamma^0 E_M^2\gamma^0$. We can bound how much $E_M$ is far from being a reflection by using \cref{eq:BE_polynomial_approx_sign}:
\begin{equation}
\begin{split}
  \|E_M^2 - \mathbb{1} \| &= \|(E_M - \varepsilon(h_w))(E_M + \varepsilon(h_w))\| \\
  &\leq \|E_M - \varepsilon(h_w)\|\,\|E_M + \varepsilon(h_w)\| \\
  &\leq \epsilon_e \cdot 2 = 2\epsilon_e,
\end{split}
\end{equation}
where we used the fact that $\varepsilon(\hw)^2 = \mathbb{1}$, and that, since $E_M$ is approximating the sign function via an error function, $\|E_M\|\le 1$. Therefore, we can write
\begin{equation}
\begin{split}
    \|\mathbb{1} - \hat{V}\hat{V}^\dagger\| &= \|\mathbb{1} - \gamma_0 E_M^2 \gamma_0 \| \\
    &= \|\gamma_0\gamma_0 - \gamma_0 E_M^2 \gamma_0 \| \\
    &= \|\gamma_0\left( \mathbb{1} - E_M^2 \right) \gamma_0 \| \\
    &\leq \|\gamma_0\| \cdot \| \mathbb{1} - E_M^2 \|\cdot \| \gamma_0 \| \\
    &\leq 1\cdot 2\epsilon_e \cdot 1 = 2\epsilon_e
\end{split}
\end{equation}

This means that the bound to the commutation relation is:
\begin{equation}
    \|[\hat{\gamma}_5,\,\hat{h}]\|
  = \|1 - \hat{V}\hat{V}^\dagger\|
  \leq 2\epsilon_e,
\end{equation}
which is \cref{eq:GW_rel_violation} of the main text, and one can see that the chiral-symmetry violation is controlled directly by the \ac{QSP} approximation error $\epsilon_e$.

\end{document}